\documentclass[12pt]{amsart}
\usepackage[margin=1.2in]{geometry}
\title[dynamical sampling for burst-like  forcing terms]{Predictive algorithms in dynamical sampling for burst-like  forcing terms }

\usepackage{hyperref}
\author[ ]{}

\date{\vspace{-5ex}}
\usepackage{amsthm, amsfonts, mathrsfs}
\usepackage{ dsfont }
\usepackage{amssymb}
\usepackage{enumerate}
\usepackage{graphicx}
 \usepackage[position=top]{subfig}
\usepackage{float}
\usepackage{bbm}
\usepackage{amsmath}
\usepackage{comment}
\usepackage{hyperref}
\usepackage{listings}
\usepackage{color}
\usepackage{diagbox}
\usepackage[dvipsnames]{xcolor}
\usepackage{algorithm}
\usepackage{algorithmicx}
\usepackage[noend]{algpseudocode}
\usepackage{tikz}
\allowdisplaybreaks

\hypersetup{
    colorlinks,
    citecolor=blue,
    filecolor=blue,
    linkcolor=blue,
    urlcolor=blue
}

\newtheorem{theorem}{Theorem}[section]

\newtheorem{corollary}[theorem]{Corollary}

\newtheorem{remark}[theorem]{Remark}

\newtheorem{example}{Example}

\newtheorem{problem}{Problem}
\newtheorem{assume}{Assumption}

\newcommand{\HH}{\mathcal{H}}

\newcommand{\R}{\mathbb{R}}
\newcommand{\Z}{\mathbb{Z}}
\newcommand{\N}{\mathbb{N}}

\newcommand{\G}{\mathcal{G}}

\newcommand {\Thresh} {Q(g)} 
\newcommand{\hatm}{\widehat{\mathfrak{m}}}

\newcommand {\la} {\langle}
\newcommand {\ra} {\rangle}

\definecolor{aacolor}{rgb}{0.05, 0.75, 1}

\definecolor{lxcolor}{rgb}{0.65, 0.15, 0.6}
\newcommand{\lx}[1]{\textcolor{black}{#1}}

\definecolor{kkcolor}{rgb}{0.0, 0.95, 0.05}

\definecolor{ikcolor}{rgb}{1, 0., 0.}
\newcommand{\ik}[1]{\textcolor{black}{#1}}

\begin{document}
	\date{}

	\author{Akram Aldroubi,
		Longxiu Huang, Keri Kornelson, and
		Ilya Krishtal}

\address{\textrm{(Akram Aldroubi)}
Department of Mathematics, Vanderbilt University, Nashville, TN 37240}
\email{akram.aldroubi@vanderbilt.edu}
\address{
\textrm{(Longxiu Huang)}
	Department of Mathematics,
	University of California, Los Angeles, CA 90095
	}
\email{huangl3@math.ucla.edu}

\address{\textrm{(Keri Kornelson)} Department of Mathematics, University of Oklahoma, Norman, OK 73019
}
\email{kkornelson@ou.edu}

\address{
\textrm{(Ilya Krishtal)}
	Department of Mathematical Sciences,
	Northern Illinois University, DeKalb, IL 60115
	}
\email{ikrishtal@niu.edu}

\keywords{Sampling Theory, Forcing, Frames,
Reconstruction, Semigroups, Continuous Sampling }
\subjclass [2010] {46N99, 42C15,  94O20}

\maketitle
\begin{abstract}
In this paper, we consider the problem of recovery of a burst-like  forcing term in an initial value problem (IVP) in the framework of dynamical sampling.  
We introduce an idea of using two particular classes of samplers that  allow one to predict the solution of the IVP over a time interval without a burst. This leads to two different algorithms that stably and accurately approximate  the burst-like forcing term even in the presence of a 
measurement acquisition error and a
large background source.   
\end{abstract}

\section{Introduction}
\subsection{Motivation.} 
In this paper, we present and study two algorithms that allow a machine to detect an unusual occurrence in an otherwise smooth environment. We will refer to these unusual occurrences as bursts. In reality, a burst can be an earthquake, a jumping fish, an explosion, a lie on a polygraph test, or an omission or insertion in a deep fake video. The algorithms are designed for continuous monitoring of the environment and will detect bursts one-by-one, assuming that there is a known (possibly very small) gap between any two consecutive bursts. (If more than one burst occurs within a small time interval, the algorithms will treat them as a single burst with a shape determined, roughly, as a superposition of the shapes of the real bursts.)

The first of the two algorithms utilizes discrete time samples of the environment and declares the time of the detected  burst to be  the mid-point between the times of the corresponding samples. The second algorithm uses specific weighted averages (Fourier coefficients) of continuous time samples and establishes the time of the burst more accurately. As for the shape of the burst, both algorithms recover it with an error controlled by a measure of smoothness of the environment, measurement acquisition error, and the time step of the algorithm chosen by the user.

The key feature of both presented algorithms is the predictive nature of the design of the samples. We assume that the monitored environment evolves under the action of a known physical process. This allows us to use the generator of the process or the operator semigroup describing it in modeling the sampling devices. Then the measurements made at a previous time step can be used to estimate the current measurements under the assumption that no burst occurred between them. Thus, comparing the estimated current measurements with the actual current measurements, one can reasonably accurately determine if a burst occurred or not.   

This work is motivated by the problem of isolating localized source terms  considered in \cite {MBD15,MBD17}.  Some real world inspiration was provided by the problems of identifying jumping fish from a video recording of the surface of a pond and transcribing the score of a musical piece from the spectrogram of its performance.

 \subsection{Problem setting}
We consider the problem of recovering the ``burst-like'' portion $f=f(x,t)$ of an unknown source term $F = f +\eta$ from space-time samples   
of a function $u=u(x,t)$ that evolves in time due to the action of a known evolution operator $A$ and the forcing function $F$. The variable $x \in \R^d$  is the ``spatial''  variable, while $t\in\R_+$ represents time. For each fixed $t$, $u(\cdot,t)$ can be  viewed as a vector $u(t)$ in a Hilbert space $\HH$ of functions on a subset of $\R^d$. With this identification we get the following abstract initial value problem:
\begin{equation}\label{DFM}
	\begin{cases}
	\dot{u}(t)=Au(t)+F(t)\\
	u(0)=u_0,
	\end{cases}
	\quad t\in\mathbb R_+,\ u_0\in\HH.
\end{equation} 
Above  $\dot{u}: \R_+\to\HH$ is the time derivative of $u$, $F:\mathbb{R}_+\rightarrow\HH$ is  a forcing (or source) term, and $A: D(A)\subseteq \HH\to\HH$ is a generator of a strongly continuous semigroup $T$.  

A prototypical example arises when  $A=\Delta$ is the Laplacian operator on Euclidean space. For this case, $F$ represents the unknown ``heat source'' a  portion of which we seek  to  recover from the space-time samples of the temperature $u$. 

In this paper, we only consider sources of the form $F = f+\eta$, where $\eta$ is a Lipschitz continuous background source and $f$ is a  ``burst-like'' forcing term given by
\begin{equation}\label{STFC}
	f(t)=\sum_{j=1}^{N}f_j\delta(t-t_j),
\end{equation}
for some unknown $N\in \N$, with $0<t_1\ldots< t_N$ 
and $f_j\in V$.  Here, the set $V$ is a 
 subset of $\HH$ and $\delta$ is the Dirac delta-function.  {We call each $t_j$ the \textit{time} of burst $j$ and $f_j$ the \textit{shape} of the burst.}  


With the notation given above, the problem studied in this paper can be stated as follows. 
\begin{problem}\label {Prob}
 Design a (finite or countable) set of samplers $\G\subseteq \HH$ and an algorithm that allow one to stably and accurately approximate any $f$ of the form \eqref {STFC} from the samples obtained from values of the measurement function $\mathfrak m$ given by  
 \begin{equation}\label{measurements}
	\mathfrak m(t,g) = \left\langle u(t),g \right\rangle +\nu(t,g),\ t\ge 0,\ g\in\G,
\end{equation}
where $\nu$ is the measurement acquisition noise. 
\end{problem}

In our recent papers (see, e.g., \cite{ACCMP17, ADK13, ADK15, AGHJKR21, AHKLLV18} and the references therein)  we considered the problem of recovering  a vector $f\in \HH$ from the samples $\{\left\langle A^n f, g\right\rangle: g\in\G,\, n = 0, 1, \ldots, L\}$ and its continuous time analog \cite{aldroubi2019frames}. We wish to utilize some of those ideas for Problem \ref {Prob}, thus putting it into the general framework of dynamical sampling. 
But first, we would like to make a few observations and assumptions that will let us explain the contributions of this paper more clearly.

The algorithms we design find the burst times $t_j$ and shapes $f_j$ one at a time from a stream of measurements. These algorithms rely on a known fixed minimal separation $\gamma$ between the time of each burst. The value $\gamma$ provides us with an upper bound on the time step $\beta$ {between samples} in the data stream of the algorithms.

\begin{assume}\label{as4}
There is a known $\gamma > 0$ such that $t_{j+1}-t_j > \gamma$  for all $j = 1, \ldots, N-1$. The time step $\beta$ is assumed to satisfy $\beta < \frac\gamma 3$ for Algorithm \ref{alg:direct} and $\beta < \frac\gamma 6$ for Algorithm \ref{alg2}.
\end{assume}


Hereinafter,  we denote  by $T$ the semigroup generated by the operator $A$ in \eqref{DFM}. The first part of Problem \ref{Prob} is the design of the set of samplers $\G$. The following example illustrates some difficulties one could face with $\G$.   The example uses properties of semigroups found in Section \ref{tool}.   

\begin{example}\label{ex:tiunknown} 
It may happen that one cannot uniquely determine a source term of the form   \eqref{STFC}.  In this example, let $T$ be the semigroup of translations acting on $\mathcal H = L^2(\mathbb{R})$, i.e. $[T(t)f](x)=f(x-t)$.  

Let $f_1 =\chi_{[0,1]}-\frac12 \chi_{[2,3]}$ and $\tilde{f}_1 = \chi_{[2,3]}$ in $L^2(\mathbb{R})$. We will consider the two burst-like source terms $f= f_1\delta(\cdot)$ and $\tilde{f}=\tilde{f}_1\delta(\cdot-2)$.  We create initial value problems in the form of \eqref{DFM} in the simplest case, where $u_0 = 0$ and there is no background source.  Looking ahead to Equation \eqref{solburst}, these problems have solutions $u$ and $\tilde{u}$, respectively, where

$$u(t) = u(\cdot, t) = T(t) f_1  \quad \mbox{ and}  \quad \tilde{u}(t) = \tilde{u}(\cdot,t) = 
\chi_{[2,\infty)}(t)T(t-2)\tilde{f}_1, \quad  t \geq 0.
$$

Let  $g=\chi_{[1,2]}+2\chi_{[3,4]}$ be a single measurement function. The measurements of the  $u$ and $\tilde{u}$ against $g$ are, respectively, $\langle u(t), g \rangle$ and $\langle \tilde{u}(t),g \rangle$, for all $t\ge 0$.   We see, however, that these measurements of $u$ and $\tilde{u}$ match for all $t\geq 0$:

$$ \langle u(t), g\rangle = \langle \tilde{u}(t), g \rangle = 
\begin{cases} 0, & t \leq 2; \\ 2t-4, & 2 < t\leq 3; \\ 8-2t, & 3 < t < 4; \\ 0, & t \geq 4.    \end{cases} 
$$
Thus, we cannot distinguish between solutions coming from the distinct forcing terms $f_1$ and $\tilde{f}_1$ using $g$ alone for the measurements. Figure \ref{fig:example} illustrates $u(t), \tilde{u}(t), \textrm{ and } g$ for selected values of $t$.

\begin{figure}[h!]
\begin{tikzpicture}
\draw[-stealth, black, thick] (0,-1.5) -- (0,0.5);
\node at (0.8,0.8) {\small \; $t=0$};

\filldraw[color=red,fill=red!5,thick](1.5,-1)rectangle(2,0);
\filldraw[color=red,fill=red!5,thick](.5,-1)rectangle(1,-.5);
\filldraw[color=blue,fill=blue!5,thick](0,-1)rectangle(.5,-.5);
\filldraw[color=blue,fill=blue!5,thick](1,-1)rectangle(1.5,-1.25);
\draw[-stealth, black, thick] (-0.5,-1) -- (2.2,-1);
\node at (1.75, 0.25){\Small \textcolor{red}{$g$}};
\node at (.3,-0.25){\Small \textcolor{blue}{$u$}};
\node at (2.4, -1){\Small $x$};

\begin{scope}[shift={(3.5,0)}];
\draw[-stealth, black, thick] (0,-1.5) -- (0,0.5);
\node at (0.8,0.8) {\small \; $t=0.5$};
\filldraw[color=red,fill=red!5,thick](1.5,-1)rectangle(2,0);
\filldraw[color=red,fill=red!5,thick](.5,-1)rectangle(1,-.5);

\filldraw[color=blue,fill=blue!5,thick](0.25,-1)rectangle(.5,-.5);
\filldraw[color=purple,fill=purple!10,thick](0.5,-1)rectangle(0.75,-.5);
\filldraw[color=blue,fill=blue!5,thick](1.25,-1)rectangle(1.75,-1.25);
\draw[-stealth, black, thick] (-0.5,-1) -- (2.2,-1);
\node at (1.75, 0.25){\Small \textcolor{red}{${g}$}};
\node at (.4,-0.25){\Small \textcolor{blue}{$u$}};
\node at (2.4, -1){\Small $x$};
\end{scope}

\begin{scope}[shift={(7,0)}]
\draw[-stealth, black, thick] (0,-1.5) -- (0,0.5);
\node at (0.8,0.8) {\small \; $t=2$};
\filldraw[color=red,fill=red!5,thick](1.5,-1)rectangle(2,0);
\filldraw[color=red,fill=red!5,thick](.5,-1)rectangle(1,-.5);

\filldraw[color=blue,fill=blue!5,thick](1.0,-1)rectangle(1.5,-.5);
\filldraw[color=blue,fill=blue!5,thick](2,-1)rectangle(2.5,-1.25);
\draw[-stealth, black, thick] (-0.5,-1) -- (2.7,-1);
\node at (1.75, 0.25){\Small \textcolor{red}{${g}$}};
\node at (1.25,-0.25){\Small \textcolor{blue}{$u$}};
\node at (2.9, -1){\Small $x$};

\end{scope}

\begin{scope}[shift={(11,0)}]
\draw[-stealth, black, thick] (0,-1.5) -- (0,0.5);
\node at (0.8,0.8) {\small \; $t=2.5$};
\filldraw[color=red,fill=red!5,thick](1.5,-1)rectangle(2,0);
\filldraw[color=red,fill=red!5,thick](.5,-1)rectangle(1,-.5);

\filldraw[color=blue,fill=blue!5,thick](1.25,-1)rectangle(1.5,-.5);
\filldraw[color=purple,fill=purple!10,thick](1.5,-1)rectangle(1.75,-.5);
\filldraw[color=blue,fill=blue!5,thick](2.25,-1)rectangle(2.75,-1.25);
\draw[-stealth, black, thick] (-0.5,-1) -- (3.0,-1);

\node at (1.75, 0.25){\Small \textcolor{red}{${g}$}};
\node at (1.35,-0.25){\Small \textcolor{blue}{$u$}};
\node at (3.2, -1){\Small $x$};

\end{scope}

\begin{scope}[shift={(0,-3)}]
\draw[-stealth, black, thick] (0,-1.5) -- (0,0.5);
\node at (0.8,0.8) {\small \; $t=0$};

\filldraw[color=red,fill=red!5,thick](1.5,-1)rectangle(2,0);
\filldraw[color=red,fill=red!5,thick](.5,-1)rectangle(1,-.5);
\draw[-stealth, black, thick] (-0.5,-1) -- (2.2,-1);
\node at (1.75, 0.25){\Small \textcolor{red}{${g}$}};

\node at (2.4, -1){\Small $x$};

\end{scope}

\begin{scope}[shift={(3.5,-3)}]
\draw[-stealth, black, thick] (0,-1.5) -- (0,0.5);
\node at (0.8,0.8) {\small \; $t=0.5$};

\filldraw[color=red,fill=red!5,thick](1.5,-1)rectangle(2,0);
\filldraw[color=red,fill=red!5,thick](.5,-1)rectangle(1,-.5);
\draw[-stealth, black, thick] (-0.5,-1) -- (2.2,-1);
\node at (1.75, 0.25){\Small \textcolor{red}{${g}$}};

\node at (2.4, -1){\Small $x$};
\end{scope}

\begin{scope}[shift={(7,-3)}]
\draw[-stealth, black, thick] (0,-1.5) -- (0,0.5);
\node at (0.8,0.8) {\small \; $t=2$};

\filldraw[color=red,fill=red!5,thick](1.5,-1)rectangle(2,0);
\filldraw[color=red,fill=red!5,thick](.5,-1)rectangle(1,-.5);

\filldraw[color=blue,fill=blue!5,thick](1,-1)rectangle(1.5,-.5);
\draw[-stealth, black, thick] (-0.5,-1) -- (2.7,-1);
\node at (1.75, 0.25){\Small \textcolor{red}{${g}$}};
\node at (1.25,-0.25){\Small \textcolor{blue}{$\tilde{u}$}};
\node at (2.9, -1){\Small $x$};

\end{scope}

\begin{scope}[shift={(11,-3)}]
\draw[-stealth, black, thick] (0,-1.5) -- (0,0.5);
\node at (0.8,0.8) {\small \; $t=2.5$};

\filldraw[color=red,fill=red!5,thick](1.5,-1)rectangle(2,0);
\filldraw[color=red,fill=red!5,thick](.5,-1)rectangle(1,-.5);

\filldraw[color=blue,fill=blue!5,thick](1.25,-1)rectangle(1.5,-.5);
\filldraw[color=purple,fill=purple!10,thick](1.5,-1)rectangle(1.75,-.5);
\draw[-stealth, black, thick] (-0.5,-1) -- (3.0,-1);

\node at (1.75, 0.25){\Small \textcolor{red}{${g}$}};
\node at (1.35,-0.25){\Small \textcolor{blue}{$\tilde{u}$}};
\node at (3.2, -1){\Small $x$};

\end{scope}
\end{tikzpicture}

\caption{\footnotesize In the figures shown here, the measurement function $g$ is shown in red.  In the first row of images, we see in blue $u(t)$ respectively at $t=0, 0.5, 2, 2.5$.  In the second row, we see $\tilde{u}(t)$ shown at the same times.  Recall that $\tilde{u}(t) = 0$ for all $t< 2$, so we don't see the function pop up until $t=2$.}\label{fig:example}
\end{figure}
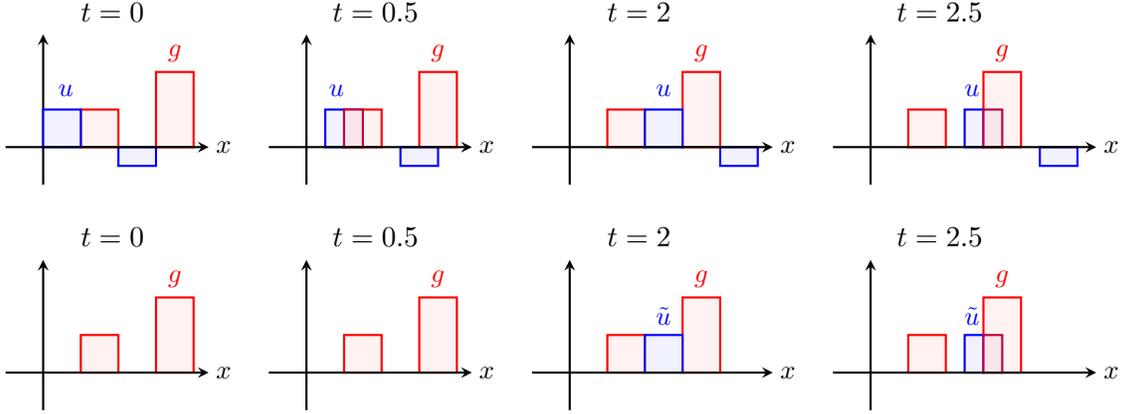

\end{example}

We break the question of sampler design into two parts. The first part is at the core of this paper and concerns the structure of the set $\G$. We express our choice of the structure in the following assumption. 

 \begin{assume}\label{as1}
For Algorithm \ref{alg:direct}, we assume that $\G = \widetilde\G \cup T^*(\beta) \widetilde\G$ for some countable (possibly, finite) set $\widetilde\G\subseteq \HH$. For Algorithm \ref{alg2}, we assume that $\G$ is of the form $\G = \widetilde\G \cup A^* \widetilde\G$ for some countable (possibly, finite) set $\widetilde\G\subseteq D(A^*)$.
\end{assume}

The second part of the sampler design question covers the need to reconstruct the shapes $f_j$ from their samples $\la f_j, g\ra$, $g\in \G$. In this paper, we address this part only superficially in the form of the following assumption. 

\begin{assume}\label{as2}
Assumption \ref{as1} holds and
the analysis map $R_{\G}: V\to \ell_\infty(\widetilde\G)$ given by
$(R_{\G}h)(g) = \left\langle h, g\right\rangle$, $h\in V$, $g\in\widetilde\G$, has a left inverse $S_\G: R_{\G}V\subseteq \ell_\infty(\widetilde\G)\to V\subseteq \mathcal H$, which is Lipschitz with a Lipschitz constant $S\ge 0$. We also let $R = \sup_{g\in\widetilde \G} \|g\|$.
\end{assume}

More specific design of  the set $\widetilde\G$ can be included, based on the application at hand. Standard methods of frame theory and compressed sensing can be used for the applications we have in mind.  

The next assumption is used to separate the burst-like term from the background source.

\begin{assume}\label{as3}
The background source $\eta$ is Lipschitz with a Lipschitz constant $L\ge 0$.
\end{assume}

Our next  assumption deals with the error in acquisition of the measurements $\la u(t), g\ra$, $t > 0$, $g\in\G$.

\begin{assume}\label{as5}
The additive noise  $\nu$ in the measurements 
\eqref {measurements}  satisfies $$\sup\limits_{t > 0,\ g\in\G}|\nu(t,g)|\le \sigma.$$
\end{assume}

In view of Assumptions \ref{as4} -- \ref{as5}, we say that an algorithm can stably and accurately approximate any $f$ of the form \eqref {STFC} if the 
recovery error in the time and shape of a single burst produced by the algorithm can be bounded above in terms of the constants $\beta$, $\gamma$, $L$, $R$, $S$, and $\sigma$.  

For Algorithm \ref{alg2}, we will use another assumption on the step size $\beta$.
 
\begin{assume}\label{as6}
 The time step $\beta$ is assumed to satisfy $L\beta < 1$ for Algorithm \ref{alg2}.
\end{assume}

\subsection{Main results}

The main contribution of this paper is the idea of the structured sampler design expressed in Assumption \ref{as1}. By choosing the structure of the set $\mathcal G$ as described,  we acquire the measurements in pairs at each sampling time 
$n\beta$, $n=0, 1, \cdots$. The second measurement in a pair predicts a value of the first measurement in the corresponding pair at the next sampling time 
{{$(n+1)\beta$}}
if no burst happens  between 
{{$n\beta$}} and {{$(n+1)\beta$}}.
This allows us to determine if a burst had occurred in the  interval
{$[n\beta, (n+1)\beta)$}.
Based on this idea and depending on the nature of available measurements, we designed two predictive algorithms. Algorithm \ref{alg:direct} utilizes discrete samples of the measurement function \eqref{measurements}, whereas Algorithm \ref{alg2} uses a small number of its Fourier coefficients over successive time intervals. Both algorithms stably and accurately 
approximate any $f$ of the form \eqref {STFC} in the presence of non-trivial Lipschitz background source and measurement acquisition error. The performance analysis of the two algorithms is presented, respectively, in Theorems \ref{pred1} and \ref{Thm: perturbation1}, and their corollaries. In particular, for each of the two algorithms, we establish guaranteed upper bounds on the recovery errors of the burst times $t_j$ and shapes $f_j$ in terms of the constants $\beta$, $\gamma$, $L$, $R$, $S$, and $\sigma$.  We emphasize that the second algorithm can detect the  bursts exactly if there is no noise and the background source is constant. These results are described in the main part of the paper, Section \ref{predsec}.

In Section \ref{secsim}, we illustrate the performance of the algorithms from Section \ref{predsec} on simple synthetic examples. We see that Algorithm \ref{alg2} typically performs much better than guaranteed by Theorem \ref{Thm: perturbation1}. The guarantees of Theorem \ref{pred1} for Algorithm \ref{alg:direct}, however, are usually \lx{sharper}. 

\subsection{IVP toolkit}\label{tool}
 
We conclude the introduction section with a reminder of basic facts from the theory of one-parameter operator semigroups and their application to solving IVPs of the form \eqref{DFM}. We refer to \cite{EN00} for more information.

A strongly continuous operator semigroup  is a map $T:\mathbb{R}_+\rightarrow B(\HH)$ (where $B(\HH)$ is the space of all bounded linear operators on $\HH$), which satisfies 
\begin {enumerate}[ (i)]
\item $T(0)=I$, 
\item $T(t+s)=T(t)T(s)$ for all $t,s\geq 0$, and
\item $\|T(t)h-h\|\rightarrow 0$ as $t\rightarrow 0$ for all $h\in \HH$.
\end {enumerate}

The operator $A$ is said to be a generator of the semigroup  $T$ if, given
\[D(A) = \left\{h\in \HH: \ \lim_{t\to 0^+} \frac1t(T(t)h-h) \mbox{ exists}\right\}, \]   $A$ satisfies 
 	\[Ah = \lim_{t\to 0^+} \frac1t(T(t)h-h),\ h\in D(A).\]
The semigroup $T$  is said to be uniformly continuous if $\|T(t)-I\|\to 0$ as $t\to 0$. In this case, $D(A) = \HH$.

 Recall \cite[p.~436]{EN00} that the (mild) solution of \eqref{DFM} can be represented  as
\begin{equation}\label{SfDFM}
	u(t)=T(t)u_0+\int_{0}^{t}T(t-\tau)F(\tau)d\tau.
\end{equation}
Substituting $F = f+\eta$ with $f$ of the form \eqref{STFC}, yields
\begin{equation}\label{solburst}
	u(t)=T(t)u_0+\sum_{t_j \le  t} T(t-t_j)f_j+\int_{0}^{t}T(t-\tau)\eta(\tau)d\tau, \quad
	t\ge 0.
\end{equation}

 \section{Predictive algorithms}\label{predsec}
 
 In this  section of the paper, we derive two algorithms for recovering the burst-like forcing term and then we perform error analyses for each. 

\subsection{Discrete sampling algorithm}\label{secalg1}

In the first  method, we acquire the following set of measurements:
\begin{equation}\label{eqn:meas_dic}
    \begin{aligned}
    \mathfrak m_n(g)=&\left\langle u(n\beta),g\right\rangle +\nu(n\beta,g),\\
    \mathfrak m_n( T^*(\beta) g)=&\left\langle u(n\beta),T^*(\beta) g\right\rangle + \nu (n\beta,T^*(\beta) g),\quad g\in\widetilde{\mathcal G}, n\in \N,
    \end{aligned}
\end{equation}
where $\beta$ is a time sampling step, and $\nu(n\beta,g),\ \nu (n\beta,T^*(\beta) g)$ represent additive noise that is assumed to be bounded according to Assumption \ref{as5}:

\[\sup\limits_{n,g}| \nu(n\beta,g)|\le \sigma, \quad \sup\limits_{n,g}| \nu (n\beta,T^*(\beta) g)|\le \sigma,
\]
for some $\sigma\ge0$.

 Thus, the totality of  
 samplers consists of the set $\mathcal G=\widetilde{\mathcal G}\cup  T^*(\beta) \widetilde{\mathcal G}$. The measurement $\mathfrak m_n( T^*(\beta) g)$ can be thought of as a predictor of the value $\left\langle u((n+1)\beta),g\right\rangle$ if no burst \lx{occurred} in $[n\beta, (n+1)\beta)$,  i.e., up to noise in the measurements and the influence of the ``slowly varying" background source, $\mathfrak m_{n+1}(g)$ should approximately   equal to $\mathfrak m_n( T^*(\beta) g)$ if no burst happened in the time interval $[n\beta, (n+1)\beta)$.

 	\begin{algorithm}[h]
\caption{The pseudo-code for approximating the time $t_*$ and the shape $f_*$,  of a burst in the time interval   $[n\beta, (n+1)\beta)$, $n\geq 1$.}
	\label{alg:direct}
	\begin{algorithmic}[1]
		\State \textbf{Goal: }Find a possible burst for  the given measurements.
		\State \textbf{Input: } The measurements: $\mathfrak m_{\ell}(g)$ and
    $\mathfrak m_{\ell}( T^*(\beta) g)$,  for $g\in\widetilde \G$, $\ell\in\{n,n+1,n+2,n+3\}$; a parameter $K >1$.
    \State \lx{Compute $\Gamma_{\ell}(g)=\mathfrak m_{\ell+1}(g)-\mathfrak{m}_{\ell}(T^*(\beta)g)$ for $\ell=n,n+1,n+2$. }
    \State \lx{Compute $\Gamma_{\ell}^{-}(g)=\Gamma_{\ell+1}(g)-\Gamma_{\ell}(g)$ for $\ell=n,n+1$. }
 
	\While{$g\in\widetilde\G$}
	\State $t_* := n\beta+\beta/2$\;
	\If{$|\Gamma_{n}^{-}(g)|\geq  K C L \beta^2\|g\|$ and $|\Gamma_{n+1}^{-}(g)|< C L \beta^2\|g\|$}
	\State ${\mathfrak f}(g) :=  -\Gamma_{n}^{-}(g)$ 
	\Else{\State ${\mathfrak f}(g) :=  0$.}
	\EndIf
 \EndWhile
 
	\State	\hspace{0.5cm}\textbf{   Output: } $t_*$ and ${\mathfrak f}(g)$ for all $g\in\widetilde\G$.   
	\end{algorithmic}
\end{algorithm}

 \begin{theorem} \label {pred1}
 Let Assumptions \ref{as4}, \ref{as1}, \ref{as3}, and \ref{as5} hold.
 Then for every $g\in \widetilde\G$,
 the term $\sum_{j=1}^{N}\langle f_j,g\rangle\delta(t-t_j) $  is well approximated by $
 \sum_{j=1}^{N}\mathfrak f_j(g)\delta(t-\tilde t_j)$ that is obtained via {\ik{successive applications of}} Algorithm \ref{alg:direct}.
 In particular, for the shape of a burst, 
we have  
\begin{equation}\label{fest1}
{|\mathfrak f_j(g)-\langle f_j,g\rangle|} \le  (K+1)C L \beta^2\|g\|+h(f_j,g,\beta)+ 4(K+1)\sigma,\ g\in\widetilde\G,
\end{equation}
where $K$ is some constant  bigger than 1 serving as the parameter in Algorithm \ref{alg:direct}, $C=\sup\limits_{t
 \in [0,\beta]}\|T(t)\|$, $\sigma=\sup\limits_{t,g} \nu(t,g)$, and  $h(f_j,g,\beta)\to 0$ as $\beta\to 0$; and, for the time of the burst, we have $|t_j-\tilde t_j| \le\beta/2$  as long as 
 $\mathfrak f_j(g)\neq 0$ for some $g\in\widetilde\G$.
\end{theorem}

Using  Assumption \ref {as2}, and an extra condition on the semigroup $T$, we get the following corollary.

\begin{corollary}\label{cor:stab_rec_dirct}
Let Assumptions \ref{as4} -- \ref{as5} hold
and assume also that   \begin{equation}
    \label{pwcond}
    \sup\limits_{0\le\alpha\le\beta}\|(T^*(\alpha)-I)g\|\le D(\beta)\|g\|,  \quad \mbox{for all}\quad g \in \widetilde \G,
\end{equation} and some $D:\mathbb R_+\to \mathbb R_+$ such that $D(\beta)\to 0$ as $\beta\to 0$. Then for any sufficiently small $\beta > 0$, the burst term $f$ of the form \eqref{STFC} is well approximated by $\tilde f (t)= \sum_{j=1}^{N}\tilde f_j\delta(t-\tilde t_j)$ that is obtained via Algorithm  \ref {alg:direct}
from the samples \eqref {measurements} ({{with $\tilde f_j = S_{\G}\mathfrak{ f}_j$}}).
 In particular, 
we have $|t_j-\tilde t_j| \le \frac \beta 2$ as long as $\tilde f_j \neq 0$ and 
\begin{equation}
    \label{Herr}
{\|\tilde f_j - f_j\|} \le S \big((K+1)C L R\beta^2+  4(K+1)\sigma+D(\beta)R\|f_j\|\big),
\end{equation}
where $R=\sup\limits_{g \in \widetilde \G} \|g\|.$
\end{corollary}

{\begin{remark}\label{bern}
If $T$ is uniformly continuous, one may let $D$ in
Corollary \ref{cor:stab_rec_dirct} be defined by $D(\beta)=k\beta$ where $k$ is some (sufficiently large) constant.
The same can be done if the samplers in $\G$ span a subspace of $\HH$ that is invariant for $T$ and such that the restriction of $T$ to this subspace is uniformly continuous. This would happen, for example, if $T$ is the translation semigroup and $\G$ is a subset of a Paley-Wiener space. More information on the Bernstein-type inequality \eqref{pwcond}, such as an explicit estimate of the constant $k$, can be found in \cite[Theorem 3.7]{BK05} and references therein.
\end{remark}
}

 \begin{proof}[Proof of Theorem \ref {pred1}] We will assume that the initial condition $u_0=0$ since we will only consider bursts for times $t>0$. 
  The expressions of  $\mathfrak m_n(g)$ and $\mathfrak m_n(T^*(\beta)g)$ defined in \eqref {eqn:meas_dic} are given by
 
 \begin{equation}\label{MN}
 \begin{aligned}
	\mathfrak m_n(g)=&\sum_{t_j< n\beta}\left\langle T(n\beta-t_j)f_j,g\right\rangle\\
	&+\int_{0}^{n\beta}\left\langle T(n\beta-\tau)\eta(\tau),g\right\rangle d\tau+\nu(n\beta,g).
	\end{aligned}
\end{equation}
and
\begin{equation}\label{MNpred}
\begin{split}
	\mathfrak m_n(T^*(\beta)g)=&\sum_{t_j< n\beta}\left\langle T(n\beta-t_j)f_j,T^*(\beta)g\right\rangle\\
	&+\int_{0}^{n\beta}\left\langle T(n\beta-\tau)\eta(\tau),T^*(\beta)g\right\rangle d\tau+\nu(n\beta,T^*(\beta)g).
	\end{split}
\end{equation}
Let $\Gamma_{n}=\mathfrak{m}_{n+1}(g)-\mathfrak{m}_{n}(T^*(\beta)g)$ be the difference between the  measurement $\mathfrak{m}_{n+1}(g)$ at time $(n+1)\beta$ and its predicted value $\mathfrak{m}_{n}(T^*(\beta)g)$ from the measurement at time $n\beta$ if no burst occurs in the interval $[n\beta,(n+1)\beta)$. Defining the function $v_n$ on $[n\beta,(n+1)\beta)$ for each $n$ as
\begin{equation} \label {VN}
    v_n(t)= \begin{cases}
             f_j\delta(t-t_j), & \text {if } t_j\in [n\beta,(n+1)\beta)\\
               \\
               0, & \text{otherwise,}
          \end{cases}.
\end{equation}
we obtain
 \begin{equation*}
 \begin{aligned}
\Gamma_{n}=& \mathfrak{m}_{n+1}(g)-\mathfrak{m}_{n}(T^*(\beta)g)\\
=&\sum_{t_j< (n+1)\beta}\left\langle T((n+1)\beta-t_j)f_j,g\right\rangle + \int_{0}^{(n+1)\beta}\left\langle T((n+1)\beta-\tau)\eta(\tau),g\right\rangle d\tau+\nu((n+1)\beta,g)\\ 
&-\sum_{t_j< n\beta}\left\langle T(n\beta-t_j)f_j,T^*(\beta)g\right\rangle-\int_{0}^{n\beta}\left\langle T(n\beta-\tau)\eta(\tau),T^*(\beta)g\right\rangle d\tau-\nu(n\beta,T^*(\beta)g) )
\\
= &\left\langle T((n+1)\beta-t_{n})v_{n},g\right\rangle +\int_{n\beta}^{(n+1)\beta}\left\langle T((n+1)\beta-\tau)\eta(\tau),g\right\rangle d\tau\\
& +\nu\left((n+1)\beta,g\right)-\nu\left( n\beta,T^*(\beta)g\right)\\
 =&\left\langle T((n+1)\beta-t_{n})v_{n},g\right\rangle +\int_{0}^{\beta}\left\langle T(\beta-\tau)\eta(n\beta+\tau),g\right\rangle d\tau+\nu\left((n+1)\beta,g\right)-\nu\left( n\beta,T^*(\beta)g\right).
 \end{aligned}
 \end{equation*}
 
In order to minimize the effect of the background source on our prediction, we compare our predictions in two consecutive time  samples  and compute $\Gamma_{n+1}-\Gamma_{n}$. Using the calculation above, we get
\begin{equation}\label{gammaDiff}
     \begin{aligned}
     &\Gamma_{n+1}-\Gamma_{n} \\
     =& \left\langle T((n+2)\beta-t_{n+1})v_{n+1},g\right\rangle+\int_{0}^{\beta}\left\langle T(\beta-\tau)\eta((n+1)\beta+\tau),g\right\rangle d\tau \\
     & -\left(\left\langle T((n+1)\beta-t_n)v_n,g\right\rangle+\int_{0}^{\beta}\left\langle T(\beta-\tau)\eta(n\beta+\tau),g\right\rangle d\tau\right) +\xi_n\\
     =&\left\langle T((n+2)\beta-t_{n+1})v_{n+1},g\right\rangle-\left\langle T((n+1)\beta-t_n)v_n,g\right\rangle \\
     &+\int_{0}^{\beta}\left\langle T(\beta-\tau)(\eta((n+1)\beta+\tau)-\eta(n\beta+\tau)),g\right\rangle d\tau+\xi_n,
     \end{aligned}
 \end{equation}
 where $\xi_n=\nu\left((n+2)\beta,g\right)-\nu\left( (n+1)\beta,T^*(\beta)g\right)-\nu\left((n+1)\beta,g\right)+\nu\left( n\beta,T^*(\beta)g\right)$.
 
 If there is no burst in $[n\beta,(n+2)\beta)$ (i.e., $v_n=v_{n+1}=0$), then the difference $|\Gamma_{n+1}-\Gamma_{n}|$ should be small, and should only depend on the sampling time  $\beta$, the 
 {samplers} $g$, the Lipschitz constant $L$ of the background source, and the noise level $\sigma$. Estimating  $|\Gamma_{n+1}-\Gamma_{n}|$ from above, we get 
  \begin{equation}\label{eqn:no_burst}
      \begin{aligned}
      |\Gamma_{n+1}-\Gamma_{n}|\leq&\left|\int_{0}^{\beta}\left\langle T(\beta-\tau)(\eta((n+1)\beta+\tau)-\eta(n\beta+\tau)),g\right\rangle d\tau\right|+|\xi_n|\\
      \leq& \int_{0}^{\beta}\left| \left\langle T(\beta-\tau)(\eta((n+1)\beta+\tau)-\eta(n\beta+\tau)),g\right\rangle\right|d\tau +4\sigma\\
      =& \int_{0}^{\beta}\left| \left\langle (\eta((n+1)\beta+\tau)-\eta(n\beta+\tau)),T^*(\beta-\tau)g\right\rangle\right|d\tau +4\sigma\\
      \leq & \int_{0}^{\beta}\|\eta((n+1)\beta+\tau)-\eta(n\beta+\tau))\|\|T^*(\beta-\tau)g\|d\tau+4\sigma\\
      \leq& C  L \beta^2\|g\|+4\sigma.
      \end{aligned}
  \end{equation}
   Using the last inequality, we will declare that a burst occurred in $[n\beta,(n+2)\beta)$ if  the value of $|\Gamma_{n+1}-\Gamma_{n}|$ is above the threshold $\Thresh=K \left({C L \beta^2\|g\|}+4\sigma\right) $ where  $K$  is some chosen number larger than 1, i.e.
   $$|\Gamma_{n+1}-\Gamma_{n}|\ge K \left({C L \beta^2\|g\|}+4\sigma\right). $$
   To decide if the burst  occurred in the interval $[n\beta,(n+1)\beta)$ or in $[(n+1)\beta,(n+2)\beta)$, we take into account the estimate $|\Gamma_{n+2}-\Gamma_{n+1}|$ between the times $(n+1)\beta$ and $(n+2)\beta$, and  use the fact that the minimal time between two bursts is $\gamma>3\beta$ (see Assumption \ref{as4}). To do this, we define the burst detector function $\mathfrak f(g)$ on $[(n+1)\beta,(n+2)\beta)$ as follows
  \begin{equation*}
    \mathfrak f(g)= \begin{cases}
             \Gamma_{n+1}-\Gamma_{n}, & \text{if }|\Gamma_{n+1}-\Gamma_{n}|\ge \Thresh ~\mbox{ and } ~|\Gamma_{n+2}-\Gamma_{n+1}|\ge  \Thresh;\\
               \\
               0, & \text{Otherwise.}
          \end{cases}.
\end{equation*}
To see how the function $\mathfrak f(g)$ behaves, 
we compute the difference \[\mathfrak f(g)- \left\langle T((n+2)\beta-t_{n+1})v_{n+1},g\right\rangle,\] where $v_{n+1}\ne 0$ when a burst occurs in the interval $[(n+1)\beta,(n+2)\beta)$ and $v_{n+1}= 0$ otherwise as in \eqref {VN}.  Using \eqref {gammaDiff}, we get
\begin{equation*}
\begin{aligned}
    &\mathfrak f(g)- \left\langle T((n+2)\beta-t_{n+1})v_{n+1},g\right\rangle = \\
     & \begin{cases}-\left\langle T((n+1)\beta-t_n)v_n,g\right\rangle 
     +\int_{0}^{\beta}\left\langle  T(\beta-\tau)(\eta((n+1)\beta+\tau)-\eta(n\beta+\tau)),g\right\rangle d\tau +\xi,& \\
   \hspace{40mm}\text{~~~~if }|\Gamma_{n+1}-\Gamma_{n}|\ge \Thresh ~\mbox{ and } ~|\Gamma_{n+2}-\Gamma_{n+1}|\ge  \Thresh;
               \\
              - \left\langle T((n+2)\beta-t_{n+1})v_{n+1},g\right\rangle, \quad \text{Otherwise}. 
          \end{cases}
\end{aligned}
\end{equation*}
We must consider two cases:  1) when $|\Gamma_{n+1}-\Gamma_{n}|\ge \Thresh$ and $|\Gamma_{n+2}-\Gamma_{n+1}|\ge  \Thresh$; 2) $|\Gamma_{n+1}-\Gamma_{n}|< \Thresh$ or $|\Gamma_{n+2}-\Gamma_{n+1}|<  \Thresh$.

For the case where $|\Gamma_{n+1}-\Gamma_{n}|\ge \Thresh$ and $|\Gamma_{n+2}-\Gamma_{n+1}|\ge  \Thresh$, we detect a burst in $[n\beta,(n+2)\beta)$ and $[(n+1)\beta,(n+3)\beta)$. But since $t_{j+1}-t_{j}>3\beta$ by assumption, there must be a single burst that occurred in  $[(n+1)\beta,(n+2)\beta)$ and there is no burst in $[n\beta,(n+1)\beta)$. Therefore, $\left\langle T((n+1)\beta-t_n)v_n,g\right\rangle=0$ since $v_n=0$. In addition, by \eqref{eqn:no_burst}, when $|\Gamma_{n+1}-\Gamma_{n}|\ge \Thresh$ and $|\Gamma_{n+2}-\Gamma_{n+1}|\ge  \Thresh$, from \eqref  {BurstErrFct}, we get that 
\[\left|\mathfrak f(g)- \left\langle T((n+2)\beta-t_{n+1})v_{n+1},g\right\rangle \right|\leq C L\beta^2\|g\|+4\sigma.
\]

 For the second case when $|\Gamma_{n+1}-\Gamma_{n}|< \Thresh$  or $|\Gamma_{n+2}-\Gamma_{n+1}|<  \Thresh$, we can assume, without loss of generality, that $|\Gamma_{n+2}-\Gamma_{n+1}|< \Thresh$. Computing the error for this case we get
 \lx{
 \begin{equation} \label {BurstErrFct}
     \begin{aligned}
    & |\mathfrak f(g)- \left\langle T((n+2)\beta-t_{n+1})v_{n+1},g\right\rangle|\\
 =&| \left\langle T((n+2)\beta-t_{n+1})v_{n+1},g\right\rangle+(\Gamma_{n+2}-\Gamma_{n+1})-(\Gamma_{n+2}-\Gamma_{n+1})|\\
 \le&| \left\langle T((n+2)\beta-t_{n+1})v_{n+1},g\right\rangle+(\Gamma_{n+2}-\Gamma_{n+1})|+|\Gamma_{n+2}-\Gamma_{n+1}|\\
 =&\left|\int_{0}^{\beta}\left\langle T(\beta-\tau)(\eta((n+2)\beta+\tau)-\eta((n+1)\beta+\tau)),g\right\rangle d\tau {\ik{+ \xi_{n+1}}}\right|+|\Gamma_{n+2}-\Gamma_{n+1}|\\
 \le& C L \beta^2\|g\|+4\sigma+\Thresh=(K+1)\Big(C L \beta^2\|g\|+4\sigma\Big).
     \end{aligned}
 \end{equation}
 }
 Using the last inequality,  we can now estimate $|\mathfrak f(g)-\langle v_{n+1},g\rangle|$ by
 \begin{equation}
     \begin{aligned}
     &|\mathfrak f(g)-\langle v_{n+1},g\rangle|\\
     \le&|\mathfrak f(g)- \left\langle T((n+2)\beta-t_{n+1})v_{n+1},g\right\rangle|\\
     &+| \left\langle T((n+2)\beta-t_{n+1})v_{n+1},g\right\rangle-\langle v_{n+1},g\rangle|\\
     \le& (K+1)\big(C L \beta^2\|g\|+4\sigma\big)+h(v_{n+1},g,\beta),
     \end{aligned}
 \end{equation}
 where $h(v_{n+1},g,\beta)=| \left\langle T((n+2)\beta-t_{n+1})v_{n+1},g\right\rangle-\langle v_{n+1},g\rangle|$. Using Property (3) of Section \ref {tool}, it follows that $h(v_{n+1},g,\beta)\to 0$ as $\beta \to 0$.
\end{proof}

\begin{proof}[Proof of Corollary \ref {cor:stab_rec_dirct}] For $t_j\in [n\beta,(n+1)\beta)$,  using the expression of $h(v_{n},g,\beta)$ derived in the proof of Theorem \ref {pred1}, Definition \eqref {VN} of $v_n$, and the assumptions of the corollary we get
\begin{equation}
    \begin{split}
        h(f_j,g,\beta)=&| \left\langle T((n+1)\beta-t_n)f_j,g\right\rangle-\langle f_j,g\rangle|\\
=&|\left\langle f_j,\big(T^*((n+1)\beta)-t_n)-I\big)g\right\rangle| 
\le D(\beta)\|g\|\|f_j\| \le D(\beta)R\|f_j\|.
    \end{split}
\end{equation}
Thus, using \eqref {fest1} and Assumption \ref{as2} we obtain
\[
{\|\tilde f_j - f_j\|} \le S \big((K+1)C L R\beta^2+   D(\beta)R\|f_j\|\big)
\]
and the result is proved.
\end{proof}

 \subsection{The Prony-Laplace algorithm}\label{secalg2}
 
 In this section, we describe the second predictive algorithm for approximating the burst-like portion $f$ of the forcing term $F$ in the IVP \eqref{DFM}. In contrast with the case of the first algorithm, here we use average samples of the measurement function \eqref{measurements}, which can be thought of as discrete samples of its short-time Fourier transform. 
 The idea of the
 algorithm is based on the Laplace transform \cite{ABHN11} and Prony's methods \cite{peter2013generalized}.

 Similar to the first algorithm, the predictive nature of this one is also manifested in a specific choice of
 the sampling set $\G$. However, this time, we use the generator $A$ rather than the semigroup $T$ (see Assumption \ref{as1}).

Under Assumptions \ref{as4} -- \ref{as6}, 
we will utilize the measurements of the form
 \begin{equation}\label{Samples}
\hatm_{k\ell}(g)=\int_{(\ell-1)\beta}^{(\ell+1)\beta} e^{-\frac{\pi i k}\beta t}\left\langle u(t), \left(\frac{\pi i k}\beta I- A^*\right)g\right\rangle dt + {\nu_{k \ell}(g)}
\end{equation}
for $\ k \in \{0, 1, 2\},\ \ell\in\N,\ g\in\widetilde\G$. 

Our goal is to provide a good approximation of the signal $f$ of the form \eqref{STFC} given the measurements \eqref{Samples}. Clearly, these samples can be easily obtained from the measurements \eqref{measurements} if those are given; 
in this case,
 \begin{equation}
     \label{nukl}
     \nu_{k\ell}(g) =  \int_{(\ell-1)\beta}^{(\ell+1)\beta} e^{-\frac{\pi i k}\beta t} \left[\frac{\pi i k}\beta\nu(t, g) - \nu(t, A^*g)\right]dt, \ k \in \{0, 1, 2\},\ \ell\in\N,\ g\in\widetilde\G.
 \end{equation}
Observe that from \eqref{nukl} and Assumption \ref{as5} we get
\begin{equation}
    \label{nuklest}
    |\nu_{k\ell}(g)| \le 2(\pi k + \beta)\sigma \le 2(2\pi + \beta)\sigma =: \tilde\sigma, \ k \in \{0, 1, 2\},\ \ell\in\N,\ g\in\widetilde\G.  
\end{equation}

Our goal will be achieved once the following theorem is proved.

 	\begin{algorithm}[h]
	\caption{The pseudo-code for approximating the time $t_*$ and the shape $f_*$,  of a burst in the time interval  $[(\ell-\frac23)\beta, (\ell+\frac{17}3)\beta]$. 
	}
	\label{alg2}
	\begin{algorithmic}[1]
		\State \textbf{Goal: }Find a possible burst given two triples of measurements per $g\in\widetilde\G$.
		\State \textbf{Input: }
		$\hatm_{k(\ell+j)}(g)$, $k \in \{0, 1, 2\}$, $j=\lx{ \{\ell\} }+\{0,1,2,3,4,5\}$, $g\in\widetilde\G$.
	\State Let $t_* = \ell -1$.	
	\For{$g\in\widetilde\G$}
	\State Let $t(g) = \ell-1$ and $\mathfrak f(g) =0$.
	\State Let \lx{ $\Gamma_{kj}  =  \hatm_{kj}(g) + (-1)^{j} \hatm_{(k+1)j}(g)$},  $k \in \{0, 1\}$, $j\in \lx{ \{\ell\}+}\{0,1,2,3,4,5\}$. 
	\State Let 
	$\Delta_{kj}=\Gamma_{kj}-\Gamma_{k(j+2)}$,  $k\in\{0,1,2\}$, {$j\in \lx{ \{\ell\}+}\{0,1,2,3\}$}. 
	\For{$j\in \lx{\{\ell\}+}\{0,1,2,3,4,5\}$}
		\State Set $\chi_{j} = 1$.
	\EndFor
	\For{$j\in \lx{ \{\ell\}+}\{0,1,2,3,4,5\}$}
		\If{$\min\{|\Delta_{0j}|, |\Delta_{1j}|, |\Delta_{0j}+\lx{(-1)^{j}}\Delta_{1j}|\}\leq \frac{64}{\pi}\beta\|g\|+16 \max\left\{ \tilde\sigma,\sqrt{\tilde\sigma}   \right\}$ 
		} 
	    \State $\chi_{j} := 0$, $\chi_{j+2}:=0$. 
	    \EndIf
	\EndFor
		\For{$j\in \lx{ \{\ell\}+}\{0,1,2,3\}$}
		\If{$\chi_{j} \chi_{j+2} =1$}
	    \State $\chi_{j} := 0$.
	    \EndIf
	\EndFor
	\For{$j\in \lx{ \{\ell\}+}\{0,1,2,3,4,5\}$}
		\If{$\chi_{j} = 1$}
		\State	$t(g) = \ell\beta-\frac\beta\pi\arg\frac{\Delta_{1j}}{\Delta_{0j}}$,
		\If{$|t(g)-\ell\beta|<\frac23\beta$}
		\State	$\mathfrak f(g) := 
		\frac{\Delta_{0j}^2}{\Delta_{0j}+\lx{(-1)^{j}} \Delta_{1j}}$.
		\EndIf
		\EndIf
	\EndFor	
	\EndFor
	\State Compute $t_*$ by averaging all $t(g)$ that are bigger than $\ell-1$.
	\State Let $f_* = S_{\G}\mathfrak{ f}$.
	\State	\textbf{Output: } $t_*$ and $f_*$.   
	\end{algorithmic}
\end{algorithm}

{\begin{theorem}\label{Thm: perturbation1}
 Suppose Assumptions \ref{as4} -- \ref{as6} hold.  
 Then  
 any burst term $f$ of the form \eqref{STFC} is well approximated by $\tilde f (t)= \sum_{j=1}^{N}\tilde f_j\delta(t-\tilde t_j)$ that is obtained via Algorithm   
 \ref{alg2}
from the samples \eqref{Samples}.
 In particular, we have
\[|t_j-\tilde t_j| \leq\frac{1}{3}\left(L\beta^2+ {\beta}\cdot\min\left\{1,\sqrt{\tilde\sigma}\right\}\right)\quad \mbox{as long as}\quad \tilde f_j \neq 0,
\]
and 
\begin{equation}
    \label{fest}
{\|f_j - \tilde f_j\|} \le SR\cdot \max\left\{\frac{5}{3}\|f_j\| (L\beta+\min\{1,\sqrt{\tilde\sigma}\}),   \frac{48}\pi L\beta^2+\frac{12\tilde\sigma}{R}\right\},
\end{equation}
where $\tilde\sigma$ is defined in \eqref{nuklest}.
\end{theorem}}

From \eqref{fest}, the following result on exact reconstruction is immediate.

  \begin{corollary}
    If $L=\sigma=0$, then we have  
    $\tilde{f}=f$.
  \end{corollary}

From the proofs below it will be clear that a simpler algorithm can be used in the case when $L = \sigma =0$ (see Algorithm \ref{alg1}). For numerical purposes and to avoid missing a burst at times equal to even integer multiples of $\beta$, this algorithm should be applied with a time step $\beta$. 

	\begin{algorithm}[h]
	\caption{The pseudo-code for finding the time $t_*$ and the shape $f_*$ of a burst in the time interval  $(0, 2\beta)$, when $L =\sigma=0$.}
	\label{alg1}
	\begin{algorithmic}[1]
		\State \textbf{Goal: }Find a possible burst given a triple of measurements for each $g\in\widetilde\G$.
		\State \textbf{Input: }
		$\hatm_{k1}(g)$, $k \in \{0, 1, 2\}$, $g\in\widetilde\G$.
	\For{$g\in\widetilde\G$}
		\State Let 
	$\lx{\Gamma_{01}}=\hatm_{01}(g)-\hatm_{11}(g)$ and $ \lx{\Gamma_{11}}=\hatm_{11}(g)-\hatm_{21}(g)$.
		\If{$\min\{\lx{|\Gamma_{01}-\Gamma_{11}|}, \lx{|\Gamma_{11}|}\}= 0$}
	\State  ${\mathfrak f}(g) :=0$.
			\Else 
		\State	$t_* := \frac\beta{\pi}  \lx{\arg \frac{\Gamma_{01}}{\Gamma_{11}}}$,
		\State	${\mathfrak f}(g) :=  \lx{\frac{\Gamma_{01}^2}{\Gamma_{01}-\Gamma_{11}}}$.
		\EndIf
	\EndFor
	\State Set $f_* = S_{\G} {\mathfrak f}$.
	\If{$f_*= 0$}
	\State  $t_* :=0$.
	\EndIf
	\State	\textbf{Output: } $t_*$ and $f_*$.   
	\end{algorithmic}
\end{algorithm}

The remainder of this section constitutes the proof of Theorem \ref{Thm: perturbation1}. 
  
 \subsubsection{Reduction to a Prony-type problem}
 Let us compute the coefficients of the (distributional) Fourier transforms $\mathcal F_\ell$ of the (generalized) functions in the equation of the IVP \eqref{DFM} on the intervals $\mathcal J_\ell = [(\ell-1)\beta, (\ell+1)\beta]$, $\ell\in\N$. Integrating by parts in the left-hand-side of the equation in \eqref{DFM}, we have (for $k\in\Z$ and $\ell\in\N$)
 \[
 \begin{split}
 (\mathcal F_\ell\{\left\langle \dot u(\cdot), g\right\rangle\})(k)& =  
e^{-{\pi i k}(\ell-1)}\left\langle  
h_\ell, g\right\rangle 
+  
\left(\mathcal F_\ell\left\{\left\langle u(\cdot), \frac{\pi i k}\beta g\right\rangle \right\}\right)(k), \quad g\in\HH,  
\end{split}
 \]
 where  $h_\ell \in \HH$ 
 depends on $u$,  $\ell$,  and $\beta$, but not on $k\in\mathbb Z$. More precisely, if no burst happens at the end points of the interval $\mathcal J_\ell$, we have \[h_\ell = u((\ell+1)\beta)-u((\ell-1)\beta)-\sum_{j:\, t_j\in \mathcal J_\ell} [u(t_j) - \lim_{t\to t_j^-} u(t)];\]
 a similar formula holds when a burst does happen at an end point of $\mathcal J_\ell$.

 The right-hand-side of the equation in \eqref{DFM} yields
 \[
 \begin{split}
(\mathcal{F}_\ell\{\langle Au(\cdot) &+f(\cdot)+\eta(\cdot), g\rangle \} ) (k)   =
\left(\mathcal F_\ell\left\{\left\langle u(\cdot), A^* g\right\rangle \right\} \right)(k) 
\\ & +  \sum_{j:\, t_j\in \mathcal J_\ell} \la f_j, g \ra e^{-\frac{\pi i k}\beta t_j} +  
\int_{(\ell-1)\beta}^{(\ell+1)\beta} e^{-\frac{\pi i k}\beta t}\left\langle \eta(t), g\right\rangle dt, 
\quad g\in\HH.  
\end{split}
 \]
 Combining these two equations with \eqref{Samples}, we get
 \begin{equation}\label{ga1}
 \begin{split}
     \hatm_{k\ell}(g) = \nu_{k\ell}(g) & - 
     (-1)^{k(\ell-1)}\left\langle 
     h_\ell, g\right\rangle+ \sum_{j:\, t_j\in \mathcal J_\ell} \la f_j, g\ra e^{-\frac{\pi i k}\beta t_j} \\ & + \int_{(\ell-1)\beta}^{(\ell+1)\beta} e^{-\frac{\pi i k}\beta t}\left\langle \eta(t), g\right\rangle dt,  \quad g\in\widetilde\G,\  \ell\in\N, \ k \in \Z.  
 \end{split}
 \end{equation}

 Summing equalities in \eqref{ga1} over $\ell \in \N$, we would get a noisy version of an irregular Prony (or super-resolution) problem studied, for example, in \cite{BDVMC08, peter2013generalized} or \cite{CF13, CF14}. Our problem is, on one hand, simpler than those in the literature precisely because we do not need to sum over $\ell \in \N$. On the other hand, our problem is harder because we need to take care of the extra unknown terms that come from the background source and the values of the function $u$ that were not measured.
 
 \begin{remark}
 A rigorous proof of 
 formula \eqref{ga1} that does not involve distributions  
 may be obtained via the use of the Laplace (rather than Fourier) transform. We skip the derivation as it distracts from the main purpose of this paper but keep the name of Laplace in the algorithm as a tacit acknowledgment. 
 \end{remark}
 
 \subsubsection{Derivation of Algorithm \ref{alg2}.}
 
 We need to determine the time $t_*$ and the shape $f_*$ of a possible burst in the time interval  $[(\ell-\frac23)\beta, (\ell+\frac{17}3)\beta]$. 
 Due to Assumption \ref{as4} there are at most two bursts in this interval. Our measurements utilize the values of the function $u$ in the interval $[(\ell-1)\beta,(\ell+6)\beta]$ of length $7\beta$.  
 We will treat them as two interlacing tuples of measurements, each of which covers the length of $6\beta$: $\{(\hatm_{k\ell}(g), \hatm_{k(\ell+2)}(g), \hatm_{k(\ell+4)}(g))$: $k\in\{0,1,2\}$, $g\in\widetilde\G\}$ and $\{(\hatm_{k(\ell+1)}(g), \hatm_{k(\ell+3)}(g)$, $\hatm_{k(\ell+5)}(g))$: $k\in\{0,1,2\}$, $g\in\widetilde\G\}$. At least one of these two tuples will detect a burst if it is sufficiently large and every such burst will be detected by our algorithm when it is run consecutively.
 
 We will use \eqref{ga1} as the starting point for the derivation of the algorithm
 rewriting it via a change of variables in the integral as
  \begin{equation*}\label{gkl11}
 \begin{split}
 \hatm_{k\ell}&(g) =\chi_{\ell} e^{-\frac{\pi k i}\beta t_*}\left\langle f_*, g\right\rangle+  \nu_{k\ell}(g)+\\
 & (-1)^{k(\ell-1)} \left(
 \int_0^{2\beta} e^{-\frac{\pi k i}\beta\tau}\left\langle  \eta(\tau+(\ell-1)\beta),g\right\rangle d\tau
 - \left\langle h_\ell, g\right\rangle\right), \ g\in \widetilde{\G}, \  \ell\in\N, \ k \in \Z,  
 \end{split}
  \end{equation*} 
where $\chi_{\ell}$ equals $1$ if $t_*\in [(\ell-1)\beta, (\ell+1)\beta)$ and $0$ otherwise. 
 The observation that  $h_\ell \in \HH$ does not depend on $k\in\mathbb Z$ allows us to get rid of it by introducing
 \begin{equation}\label{gakl}
\begin{split}
\Gamma_{k\ell}(g)  =&~  \hatm_{k\ell}(g) + (-1)^{\ell} \hatm_{(k+1)\ell}(g)  \\
=&~  \chi_{\ell} e^{-\frac{\pi k i}\beta t_*}\left(1+ (-1)^{\ell} e^{-\frac{\pi i}\beta t_*}\right)\left\langle f_*, g\right\rangle+(-1)^{k(\ell-1)}\widetilde\eta_{k\ell}(g)  \\ 
&~-
    \left((-1)^{k(\ell-1)}\left\langle h_\ell, g\right\rangle + (-1)^{(k+1)(\ell-1)}(-1)^{\ell}\left\langle h_\ell, g\right\rangle \right)   \\
    &~+(\nu_{k\ell}(g)-(-1)^{\ell}\nu_{(k+1)\ell}(g))  \\
=&~  \chi_{\ell} e^{-\frac{\pi k i}\beta t_*}\left(1+ (-1)^{\ell} e^{-\frac{\pi i}\beta t_*}\right)\left\langle f_*, g\right\rangle   \\
&~+(-1)^{k(\ell-1)}\widetilde\eta_{k\ell}(g)+(\nu_{k\ell}(g)-(-1)^{\ell}\nu_{(k+1)\ell}(g)),  
\end{split}
\end{equation}
where
\[
\widetilde\eta_{k\ell}(g) = 
\int_0^{2\beta} e^{-\frac{\pi k i}\beta\tau}\left(1+(-1)^\ell e^{-\frac{\pi i}\beta \tau}\right)\left\langle  \eta(\tau+(\ell-1)\beta),g\right\rangle d\tau.
\]

If there were no measurement acquisition errors or background source, we could use \eqref{gakl} to determine $t_*$ and $f_*$ under the condition that $\chi_{\ell} = 1$ and
\begin{equation}\label{thirdcond}
1+(-1)^\ell e^{-\frac{\pi  i}\beta t_*} \neq 0.
\end{equation}
(This last requirement is the reason for having to interlace two tuples of measurements; we will elaborate further below). Indeed, considering \eqref{gakl} for $k=0$ and $k=1$ we get two equations with two unknowns which can be solved explicitly as long as it does not involve division by $0$, i.e. $\Gamma_{1\ell} \neq 0$, $\Gamma_{0\ell} \neq \Gamma_{1\ell}$, and \eqref{thirdcond} holds. The result for $\ell=1$ is summarized as Algorithm \ref{alg1}.

Essentially the same idea works once the background source and the measurement acquisition errors have been accounted for.
To do so, we use the same trick as in the derivation of Algorithm \ref{alg:direct}. In particular, we let
\begin{equation}\label{Delta}
\begin{split}
    \Delta_{k\ell}(g) &= \Gamma_{k\ell}(g)-\Gamma_{k(\ell+2)}(g) \\
    &= (\chi_{\ell} - \chi_{\ell+2}) e^{-\frac{\pi k i}\beta t_*}\left(1+ (-1)^{\ell} e^{-\frac{\pi i}\beta t_*}\right)\left\langle f_*, g\right\rangle
    + \epsilon_{k\ell}(g)+\mu_{k\ell}(g),
\end{split}
\end{equation}
where
\begin{equation}
    \label{mukl}
    \mu_{k\ell}(g)=(\nu_{k\ell}(g)-\nu_{k(\ell+2)}(g))-(-1)^{\ell}(\nu_{(k+1)\ell}(g)-\nu_{(k+1)(\ell+2)}(g)),
\end{equation}
and
\[
\begin{split}
    \epsilon_{k\ell}(g) &= (-1)^{k(\ell-1)}\left(\widetilde\eta_{k\ell}(g) - \widetilde\eta_{k(\ell+2)}(g)\right) \\ &=(-1)^{k(\ell-1)}\int_0^{2\beta} e^{-\frac{\pi k i}\beta\tau}\left(
    1+(-1)^\ell e^{-\frac{\pi  i}\beta\tau}
    \right)\left\langle  \eta(\tau+(\ell-1)\beta)-\eta(\tau+(\ell+1)\beta),g\right\rangle d\tau
\end{split}
\]
 is controlled by the Lipschitz constant $L$ of the background source $\eta$:
\begin{equation}\label{epsest}
|\epsilon_{k\ell} (g)| \le 2L\beta\|g\| \int_0^{2\beta}\left|
    1+(-1)^\ell e^{-\frac{\pi  i}\beta\tau}
    \right|d\tau  = \frac{16}\pi L\beta^2\|g\|.
\end{equation}

Assumption \ref{as5} allows us to control the error due to measurement acquisition,
 since \eqref{nuklest} and \eqref{mukl} imply
\begin{equation}\label{noi_mea_est}
    |\mu_{k\ell}(g)|\leq 4\tilde\sigma.
\end{equation}

The equations \eqref{Delta} (with $k=0$ and $k=1$) can now be used in place of equations \eqref{gakl} to approximate $t_*$ and $f_*$. For numerical purposes, we must not just avoid dividing by $0$ but also avoid dividing by a very small number. To this end, we introduce the following threshold.  We will write that $\Delta_{\cdot\ell}$ detects a burst if there is $g\in\widetilde \G$ such that
\begin{equation}
\label{Deltaest}
\min\left\{|\Delta_{0\ell}(g)|, |\Delta_{1\ell}(g)|, |\Delta_{0\ell}(g)+(-1)^\ell\Delta_{1\ell}(g)|\right\}\ge Q(g),
\end{equation}
where $Q(g)\geq \frac{64}\pi \beta\|g\|+16\tilde\sigma$ is some constant that depends on $g$; the lower bound for $Q(g)$ is chosen in view of \eqref{epsest}, \eqref{noi_mea_est}, and Assumption \ref{as6}.

Observe that due to Assumptions \ref{as4} and \ref{as6}  
at least one of $\Delta_{\cdot\ell}$, $\Delta_{\cdot(\ell+2)}$, $\Delta_{\cdot(\ell+4)}$ does not detect a burst. 
Therefore, if  $\Delta_{\cdot\ell}$ does detect a burst,  we can unambiguously determine whether $\chi_{\ell} = 1$ or $\chi_{\ell+2} = 1$, i.e.~which of the two possible intervals { $[(\ell-1)\beta,(\ell+1)\beta)$ or $[(\ell+1)\beta,(\ell+3)\beta)$} contains $t_*$. More precisely, if 
  $\Delta_{\cdot\ell}$ and $\Delta_{\cdot(\ell+2)}$ both detect a burst, then $\chi_{\ell+2} = 1$; if $\Delta_{\cdot\ell}$ detects a burst and $\Delta_{\cdot(\ell+2)}$ does not, then $\chi_{\ell} = 1$.
  
 Once a burst is detected, say, by $\Delta_{\cdot\ell}$, 
 we find an approximation $t(g)$ of $t_*$ from the equality 
 \begin{equation}\label{etg}
     e^{-\frac{\pi i}\beta t(g)} = \frac{\Delta_{1\ell}(g)}{\Delta_{0\ell}(g)}\cdot\left|\frac{\Delta_{0\ell}(g)}{\Delta_{1\ell}(g)} \right|.
 \end{equation}
 Observe that due to \eqref{Delta} for $k =0$ and $k=1$ we would have $t_* = t(g)$ in the case when $\sigma = L = 0$ and \eqref{thirdcond} holds.
In any case, we let  
 \begin{equation}\label{tg}
 t(g) = \begin{cases}
 (\ell+2\chi_{\ell+2})\beta-\frac\beta\pi\arg\frac{\Delta_{1\ell}(g)}{\Delta_{0\ell}(g)}-\beta\cdot\text{sign}\left(-\frac\beta\pi\arg\frac{\Delta_{1\ell}(g)}{\Delta_{0\ell}(g)}\right) ,~ \ell\text{~is odd}  \\
  (\ell+2\chi_{\ell+2})\beta-\frac\beta\pi\arg\frac{\Delta_{1\ell}(g)}{\Delta_{0\ell}(g)}, ~\ell\text{~is even}
 \end{cases}
 \end{equation} 
 {where  $\arg(x)\in(-\pi,\pi]$, $x\in\mathbb C$, and $\text{sign}(x)=\begin{cases}
 1, & x>0;\\
 0, & x=0;\\
 -1,& x<0.
 \end{cases}$ 
}

 Using the same reasoning, we also let 
 \begin{equation}\label{fg}
   \mathfrak f(g) = (-1)^{\chi_{\ell+2}}\frac{\Delta_{0\ell}^2(g)}{\Delta_{0\ell}(g)+(-1)^\ell\Delta_{1\ell}(g)};
   \end{equation}
 the denominator in \eqref{fg} explains a seemingly odd term in the threshold \eqref{Deltaest}. To reiterate, $t(g)$ and $\mathfrak f(g)$ defined by \eqref{tg} and \eqref{fg}, respectively, solve equations \eqref{Delta} with $k=0$ and $k=1$ in the case when $\sigma = L =0$.
 
 We now approximate $t_*$ by $\tilde t$ -- a (weighted) average of $t(g)$ and $f_*$ by $\tilde f = S_{\mathcal G}\mathfrak f$. 
 We use a weighted average to determine $\tilde t$  in order to account for a strengthened version of \eqref{thirdcond}. More precisely, we average only the values $t(g)$ for which 
 \begin{equation*}\label{cosest}
     \left|1 -  \left|\frac{\Delta_{0\ell}(g)}{\Delta_{1\ell}(g)} \right|\cdot\frac{\Delta_{1\ell}(g)}{\Delta_{0\ell}(g)}\right| =
     \left|1 - e^{-\frac{\pi  i}\beta t(g)} \right| =
     2\left|\sin\left(\frac\pi{2\beta}t(g)\right) \right| \ge 1, 
 \end{equation*}
 i.e.~when $|t(g) - (\ell+2\chi_{\ell+2})\beta| \le \frac{2\beta}3$; otherwise, the burst will be detected by $\Delta_{\cdot(\ell-1)}$, $\Delta_{\cdot(\ell+1)}$ or $\Delta_{\cdot(\ell+3)}$ (this explains how interlacing the tuples of measurements works).

The result is summarized in Algorithm \ref{alg2} and its derivation is now complete.
It is clear that applying it for $\ell\in\N$ gives an approximation for all observed bursts.
Thus, the proof of Theorem \ref{Thm: perturbation1} will be complete once we provide estimates for the approximation errors.

\subsubsection{Estimating the approximation error.}
In this section, we assume that $\Delta_{\cdot\ell}$ detected a burst (for some $\ell\in\N$) and estimate the errors of the recovery of the time $t_*$ and shape $f_*$ of the burst. 

First, let us note that if the background source $\eta$ is a constant and there are no measurement acquisition errors, we have $\epsilon_{k\ell}(g) = 0$ and $\mu_{k\ell}(g)=0$ in \eqref{Delta}, and thus the reconstruction is exact.

Secondly, let us estimate the error $|\tilde t - t_*|$. According to the idea of interlacing the tuples of measurements described above, we assume that
\begin{equation}\label{tstar}
    |t_* - (\ell+2\chi_{\ell+2})\beta| \le \frac{2\beta}3.
\end{equation}
From \eqref{Delta}, \eqref{epsest}, \eqref{noi_mea_est} and \eqref{Deltaest}, 
we have

\begin{equation}\label{fracest}
\begin{split}
\left|\frac{\Delta_{1\ell}(g)}{\Delta_{0\ell}(g)} - e^{-\frac{\pi i}\beta t_*}\right| &\leq\left|\frac{e^{-\frac{\pi i}\beta t_*}\epsilon_{0\ell}(g) -  \epsilon_{1\ell}(g)}{\Delta_{0\ell}(g)} \right|+ \left|\frac{e^{-\frac{\pi i}\beta t_*}\mu_{0\ell}(g) -  \mu_{1\ell}(g)}{\Delta_{0\ell}(g)} \right|\\
&\le \frac{32L\beta^2\|g\|+8\pi\tilde\sigma }{\pi Q(g) } 
\end{split}
\end{equation}
and 
\begin{equation}
\begin{split}
   & 1- \frac{32 L \beta^2\|g\|+8\pi\tilde\sigma }{\pi Q(g) }\le \left|\frac{\Delta_{1\ell}(g)}{\Delta_{0\ell}(g)}\right|\\
    \le&1+  \left|\frac{e^{\frac{\pi i}\beta t_*}\epsilon_{1\ell}(g)-\epsilon_{0\ell}(g)}{\Delta_{0\ell}(g)}\right|+\left|\frac{e^{\frac{\pi i}\beta t_*}\mu_{1\ell}(g)-\mu_{0\ell}(g)}{\Delta_{0\ell}(g)} \right|\\ \le& 1+ \frac{32  L \beta^2\|g\|+8\pi\tilde\sigma }{\pi Q(g) }.
\end{split}
\end{equation}
It then follows
 from \eqref{etg} that  
\begin{equation}
\begin{split}
2\left| \sin\left({\frac{\pi}{2\beta}(t(g) - t_*)}\right)\right|=& \left| e^{-\frac{\pi i}\beta t(g)} - e^{-\frac{\pi i}\beta t_*}\right| \\
=&  \left|\frac{\Delta_{1\ell}(g)}{\Delta_{0\ell}(g)}\cdot \left|\frac{\Delta_{0\ell}(g)}{\Delta_{1\ell}(g)}\right| - e^{-\frac{\pi i}\beta t_*}\right| \\
 \le &\left|\frac{\Delta_{1\ell}(g)}{\Delta_{0\ell}(g)} - e^{-\frac{\pi i}\beta t_*}\right|+ \left|1-\left|\frac{\Delta_{0\ell}(g)}{\Delta_{1\ell}(g)}\right| \right|\cdot \left|\frac{\Delta_{1\ell}(g)}{\Delta_{0\ell}(g)}\right| \\ 
 \le& \frac{32 L \beta^2\|g\|+8\pi\tilde\sigma }{\pi Q(g) }+  \left|1-\left|\frac{\Delta_{1\ell}(g)}{\Delta_{0\ell}(g)}\right| \right| \\
 \le& 2\cdot\frac{32 L \beta^2\|g\|+8\pi\tilde\sigma }{\pi Q(g) }.
\end{split}
\end{equation}
Since $Q(g)\geq  \frac{64}\pi L \beta^2\|g\|+16\tilde\sigma$ (due to Assumption \ref{as6}), we have 
\begin{equation}\label{qest}
  0\le   \frac{32 L \beta^2\|g\|+8\pi\tilde\sigma }{\pi Q(g) }\leq \frac{1}{2}.
\end{equation}
Thus, we get
\[\begin{aligned}
|t(g) - t_*| \le& \frac{2\beta}\pi\cdot \arcsin \left(\frac{32 L \beta^2\|g\|+8\pi\tilde\sigma }{\pi Q(g) }\right)\\
\le& \frac{2\beta}{3} \cdot \frac{32 L \beta^2\|g\|+8\pi\tilde\sigma }{\pi Q(g) }\\
=& \frac{2}{3} \cdot \frac{32 L \beta^3\|g\|+8\pi\beta\tilde\sigma }{\pi Q(g)}.
\end{aligned}
\]
Averaging over $g$ may only improve the result, which yields the inequality
\[
|\tilde t - t_*| \le \frac{2}{3} \cdot \frac{32 L \beta^3\|g\|+8\pi\beta\tilde\sigma }{\pi Q(g)}.
\]

Next, let us estimate the relative error of approximation of $\langle f_*, g\rangle$ when $\mathfrak f(g)\neq 0$. From \eqref{Delta} and \eqref{fg}, we get (since  $\chi_\ell-\chi_{\ell+2}=(-1)^{\chi_{\ell+2}}$)
\[
\frac{\mathfrak f(g)}{\Delta_{0\ell}(g)} =\frac{(-1)^{\chi_{\ell+2}}}{\Delta_{0\ell}(g)+(-1)^\ell\Delta_{1\ell}(g)}\left[
(-1)^{\chi_{\ell+2}}\left(1+(-1)^\ell e^{-\frac{\pi i}\beta t_*}\right)\left\langle f_*, g\right\rangle
    + \epsilon_{0\ell}(g)  + \mu_{0\ell}(g)\right],
\]
which implies
\[\left((-1)^{\ell}+ \frac{\Delta_{1\ell}(g)}{\Delta_{0\ell}(g)}\right)\left(\mathfrak f(g) -  \left\langle f_*, g\right\rangle\right)+\left(\frac{\Delta_{1\ell}(g)}{\Delta_{0\ell}(g)} - e^{-\frac{\pi i}\beta t_*}\right)\left\langle f_*, g\right\rangle =(-1)^{\chi_{\ell+2}+\ell}( \epsilon_{0\ell}(g)  + \mu_{0\ell}(g)).
\]
In view of \eqref{epsest}  and \eqref{noi_mea_est}, it follows that
\[
\begin{aligned}
&\left|\left((-1)^{\ell}+ \frac{\Delta_{1\ell}(g)}{\Delta_{0\ell}(g)}\right)\left(\mathfrak f(g) -  \left\langle f_*, g\right\rangle\right) + \left(\frac{\Delta_{1\ell}(g)}{\Delta_{0\ell}(g)} - e^{-\frac{\pi i}\beta t_*}\right)\left\langle f_*, g\right\rangle \right|\\
=&|(-1)^{\chi_{\ell+2}+\ell}( \epsilon_{0\ell}(g)  + \mu_{0\ell}(g))|\le\frac{16}\pi L\beta^2\|g\|+4\tilde\sigma.
\end{aligned}
\]
Using \eqref{fracest} and 
\begin{equation}\label{fracest2}
\begin{split}
   & \left|(-1)^\ell+\frac{\Delta_{1\ell}(g)}{\Delta_{0\ell}(g)} \right| 
   = \left|(-1)^\ell+ e^{-\frac{\pi i}\beta t_*}-e^{-\frac{\pi i}\beta t_*}+\frac{\Delta_{1\ell}(g)}{\Delta_{0\ell}(g)}\right| \\
   \ge& \left|(-1)^\ell+ e^{-\frac{\pi i}\beta t_*}\right|-\left|-e^{-\frac{\pi i}\beta t_*}+\frac{\Delta_{1\ell}(g)}{\Delta_{0\ell}(g)}\right|\\
    \ge& 1 - \frac{32 L \beta^2\|g\|+8\pi\tilde\sigma }{\pi Q(g) },
    \end{split}
\end{equation}
which holds due to \eqref{tstar},  
we get  
\begin{equation}\label{needed}
\begin{aligned}
    &\left(1- \frac{32 L \beta^2\|g\|+8\pi\tilde\sigma }{\pi Q(g) }\right)\left|\mathfrak f(g) -  \left\langle f_*, g\right\rangle\right| -  \frac{32 L \beta^2\|g\|+8\pi\tilde\sigma }{\pi Q(g) } \left|\left\langle f_*, g\right\rangle\right|\\ \le&\frac{16}\pi L\beta^2\|g\|+4\tilde\sigma.
    \end{aligned}
\end{equation}

From \eqref{Delta}, \eqref{tstar} and \eqref{qest}, we get 
\[\left|\left\langle f_*, g\right\rangle\right|\ge |\Delta_{0\ell}|-(|\nu_{0\ell}|+|\mu_{0\ell}(g)|)\geq Q(g)-\left(\frac{16}\pi L\beta^2\|g\|{+4\tilde\sigma}\right) \geq 3\cdot\left(\frac{16}\pi L\beta^2\|g\|{+4\tilde\sigma}\right) > 0.
\]
It follows from the above inequalities and \eqref{needed} that
\[
\begin{aligned}
&\left|\frac{\mathfrak f(g) -  \left\langle f_*, g\right\rangle}{\left\langle f_*, g\right\rangle}\right| \\
\le& 2\cdot \frac{32 L \beta^2\|g\|+8\pi\tilde\sigma }{\pi Q(g) }   +\left(\frac{16}\pi L\beta^2\|g\|{+4\tilde\sigma}\right)  \frac{2\pi}{ \pi Q(g)-(16 L \beta^2\|g\|+4\pi\tilde\sigma)}\\
=& \frac{4(16 L \beta^2\|g\|+4\pi \tilde\sigma)}{\pi Q(g) } +\frac{2(16 L \beta^2\|g\|+4\pi \tilde\sigma)}{\pi Q(g)-(16 L \beta^2\|g\|+4\pi\tilde\sigma)}.
\end{aligned}
\]
Using \eqref{qest}  
once again, we thus have
\[\left|\frac{\mathfrak f(g) -  \left\langle f_*, g\right\rangle}{\left\langle f_*, g\right\rangle}\right| \le  \frac{20(16 L \beta^2\|g\|+4\pi \tilde\sigma)}{3\pi Q(g) }. 
\]
Let $Q(g)=\frac{64}{\pi}\beta\|g\|+16\max\{\tilde\sigma,\sqrt{\tilde\sigma}\}$. Then if $\mathfrak f(g)\neq 0$,
\[
|\tilde t - t_*| \le \frac{2}{3} \cdot \frac{32 L \beta^3\|g\|+8\pi\beta\tilde\sigma }{\pi Q(g)}\leq \frac{1}{3}L\beta^2+ \frac{\beta}{3}\cdot\min\{1,\sqrt{\tilde\sigma}\}  
\]
and
\[
\left|\frac{\mathfrak f(g) -  \left\langle f_*, g\right\rangle}{\left\langle f_*, g\right\rangle}\right|\leq\frac{20(16 L \beta^2\|g\|+4\pi \tilde\sigma)}{3\pi Q(g) }\leq  \frac{5}{3}(L\beta+\min\{1,\sqrt{\tilde\sigma}\}).
\]
Finally, Assumption \ref{as2} yields 
\eqref{fest}
and Theorem \ref{Thm: perturbation1} follows.

\section{Simulation}\label{secsim}

 In order to evaluate our theoretical analysis for the algorithms, we compare the estimates of the bursts and the ground truth for a specific dynamical system. More specifically, we consider the burst function $f(x,t)$ of the form
\[
\begin{aligned}
f(x,t)
=& f_1(x)\delta(t-0.25)+f_2(x)\delta(t-1.5)+f_3(x)\delta(t-2.75)
\end{aligned}
\] 
with $f_1(x)= 0.35\sin(x)$, $f_2(x)=\cos(x)$, $f_3(x)=1+\sin(x)$, $t\in[0,4],x\in[0,1]$ and the dynamical system $$\frac{d u}{dt}=Au+f+\eta \text{ with } Au(x,t)=-x^2u(x,t).$$
The choice of the operator $A$ is motivated by its connection (via the Fourier transform) with the second derivative operator on a Paley-Wiener space.  Let  $g_1(x)=1$, $g_2(x)=x$, $g_3(x)=x^2$ be the sensor functions.  The measurements are generated according to \eqref{eqn:meas_dic} for Algorithm \ref{alg:direct} and \eqref{Samples} for Algorithm \ref{alg2} for some specific $\beta$. The goal is to estimate
the coefficients $\langle f_j,g\rangle$ for $j\in\{1,2,3\}$, $g(x)\in\{g_1(x),g_2(x),g_3(x)\}$ and burst time slots $\{0.25,1.5,2.75 \}$. In this simulation, we test on two types of background source $\eta$: \[\eta(x,t)=\cos(Ltx)+C \text{ and } \eta(x,t)=x\exp(-Lt)+C.\] 

  \begin{figure}[h!]
    \centering
    \includegraphics[width=0.45\textwidth]{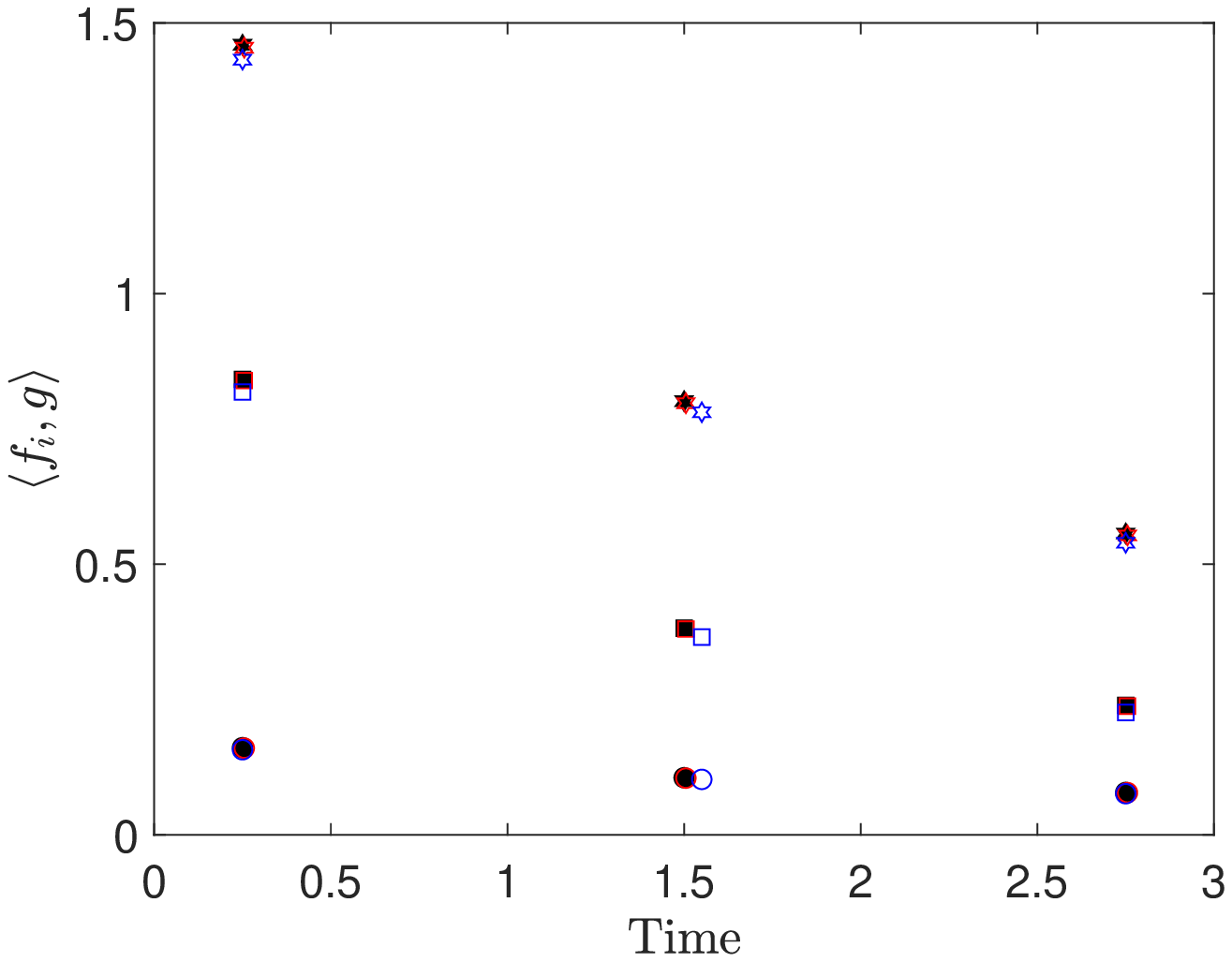}
     \includegraphics[width=0.45\textwidth]{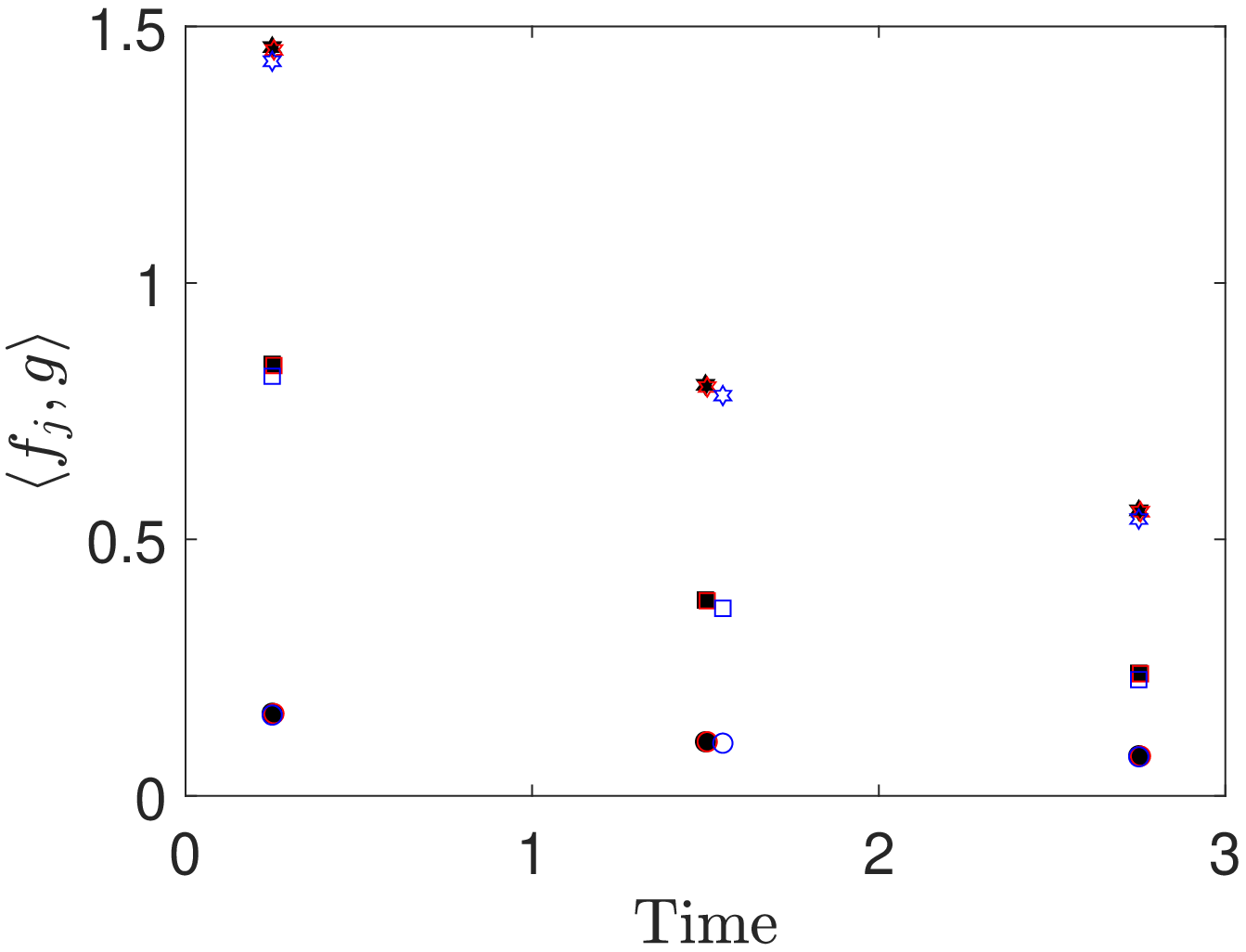}
    \caption{ \footnotesize{Plot for the bursts: The black points stand for the ground truth  of the bursts; the red points stand for the estimated bursts with the parameters $L=10^{-2},\sigma=10^{-4},\beta=10^{-2}$; the blue points stand for the estimated bursts with the parameters $L=10^{-2},\sigma=10^{-4},\beta=10^{-1}$. 
    The background sources are   $\eta(x,t) = \cos(Ltx)+C$ for  left figure and $\eta=x\exp(-Lt)+C$ for right figure.  }}
    \label{fig:coeff_time_direct}
\end{figure}

First of all, we choose some specific parameters and plot estimates and ground truth in the same figure. The estimates for Algorithm \ref{alg:direct} for some specific parameters are shown in   Figure \ref{fig:coeff_time_direct}. As the estimates for Algorithm \ref{alg2} are visually indistinguishable from the ground truth for the same parameters as the ones for Algorithm \ref{alg:direct}, we leave  out the figure for Algorithm \ref{alg2}. {From Figure \ref{fig:coeff_time_direct}, we could see that the red points are overlapping with the black points and the blue points are pretty close to the black points} which demonstrates the accuracy of our algorithms.

To further understand the influences of the parameters  $\beta$,   $L$, $\sigma$ or $\widetilde{\sigma}$ and the noise categories on our algorithms, we have done some simulations for one varying parameter and fixed others.
In our simulation, we measured the 
error on the estimates of time by computing \[\sqrt{\sum_{j=1}^{3}|t_j-\tilde{t}_j|^2}\] and  the 
error on the estimates of $\langle f_j,g\rangle$ by    \[\sqrt{\sum_{j=1}^{3}|\langle f_j,g\rangle-\mathfrak{f}_j(g) |^2}\]
for different $L$, $\beta$ and $\sigma$ or $\widetilde{\sigma}$.

\begin{figure}[!h]
\vspace{-0.08in}
\centering
\subfloat
{\includegraphics[width=0.38\linewidth]{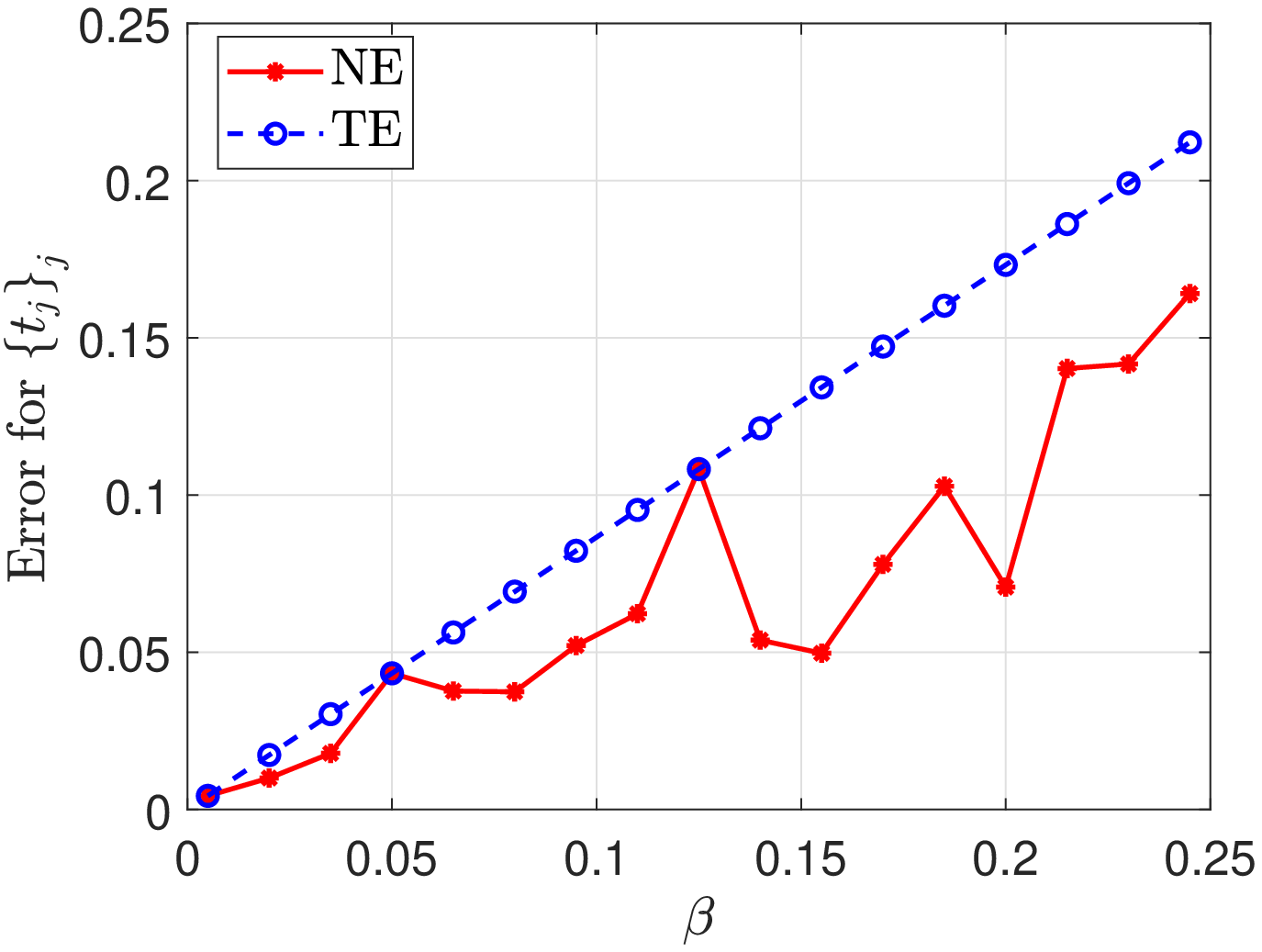} }
\subfloat
{\includegraphics[width=0.38\linewidth]{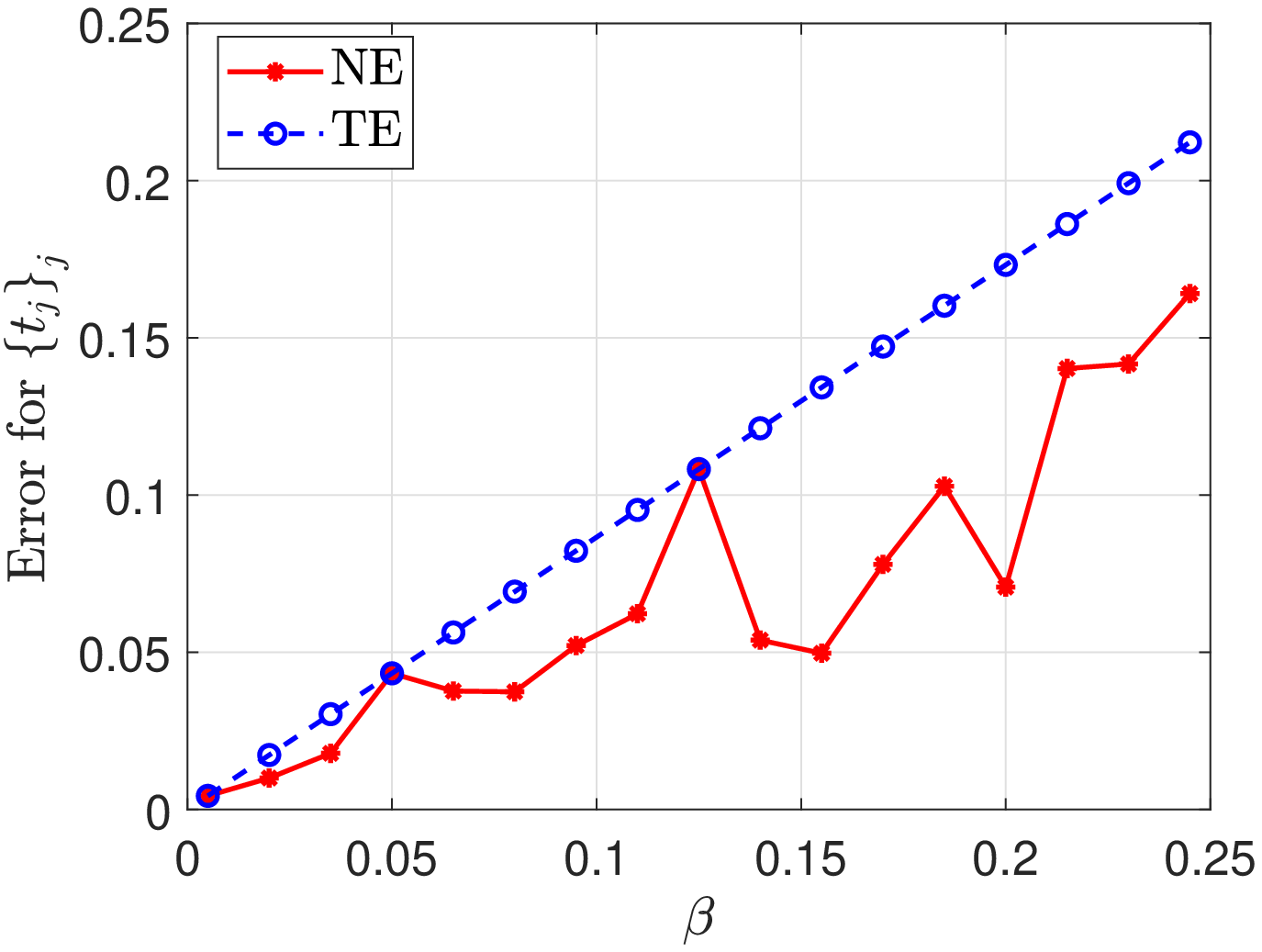} }
\hfill
\subfloat
{\includegraphics[width=0.38\linewidth]{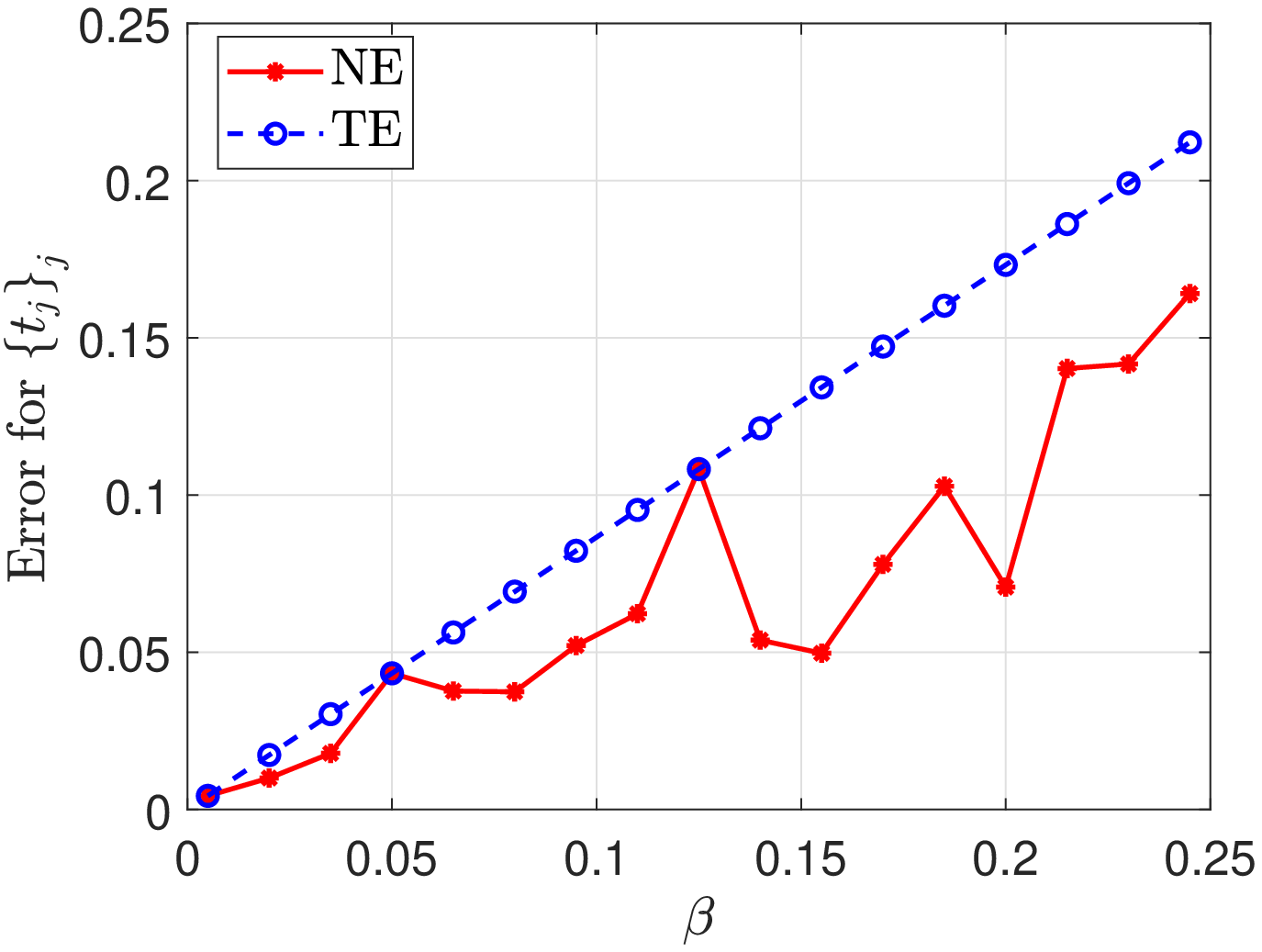}
}
\subfloat
{\includegraphics[width=0.38\linewidth]{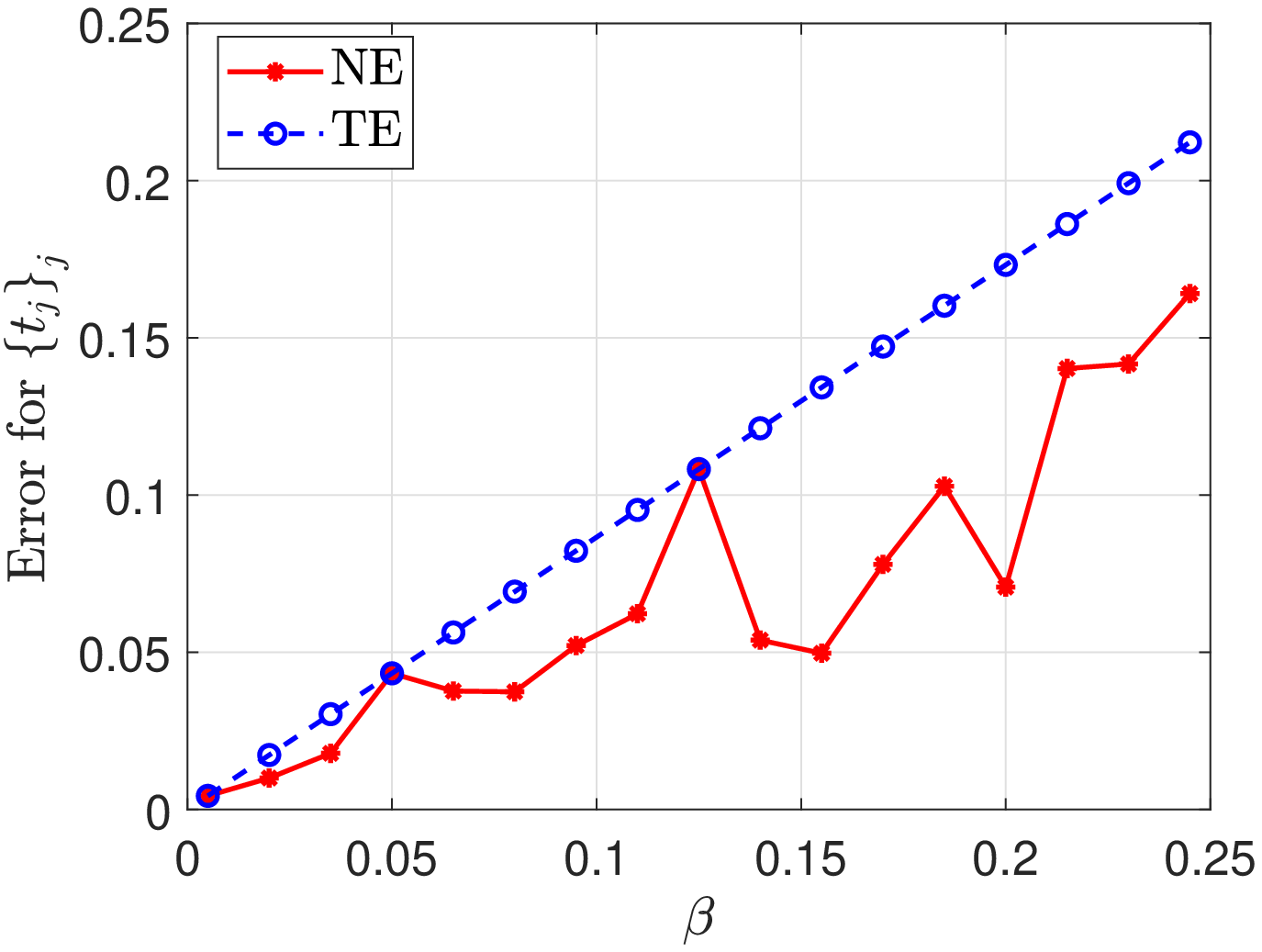}
}\hfill
\caption{ \footnotesize{The estimates of $t_j$~\textit{vs.}~$\beta$ for Algorithm \ref{alg:direct}: $L=0.01$. \textbf{First column:}   $\eta(x,t)= \cos(Ltx)+C$. \textbf{Second column:}   $\eta(x,t)= x\exp(-Lt)+C$. The noise variances $\sigma$ on the measurements are $\sigma=0,10^{-4}$ for the 1st, 2nd  row, respectively. NE and TE stand for the numerical error and the theoretical error respectively.}}
    \label{fig:time_vs_beta_direct}
\vspace{-0.05in}
\end{figure}

In Figures \ref{fig:time_vs_beta_direct}, \ref{fig:burst_vs_beta_direct}, \ref{fig:time_vs_beta_prony}, and \ref{fig:burst_vs_beta_prony}, we plot the relation between the errors on $t$ and $\langle f_j,g\rangle$ for $g(x)=x$ \textit{vs} the sampling time step $\beta$ by fixing the Lipschitz constant of the background source $L=0.01$ for different measurement variances $\sigma$ or $\widetilde{\sigma} =0,   10^{-4}$. These  figures 
show that our theoretical error bounds on the time and $\langle f_j,g\rangle$ are accurate which is close to the numerical error bound. Additionally, the performance of the estimates of the time is independent of the variance of the measurements when the burst can be detected. Meanwhile, as the variances of  measurements increase, the theoretical error bounds become worse.

For Figures \ref{fig:time_vs_lip_prony} and \ref{fig:burst_vs_lip_prony}, we fix $\beta=0.1$ and vary the Lipschitz constant $L$ of the background source $\eta(x,t)$. These figures  verify our theoretical bounds. We notice that the theoretical bound for Algorithm \ref{alg2}  provides a very accurate estimation for the time and $\langle f_j,g\rangle$ when the Lipschitz constant $L$ is small.

\begin{figure}[!h]
\vspace{-0.08in}
\centering
\subfloat
{\includegraphics[width=0.38\linewidth]{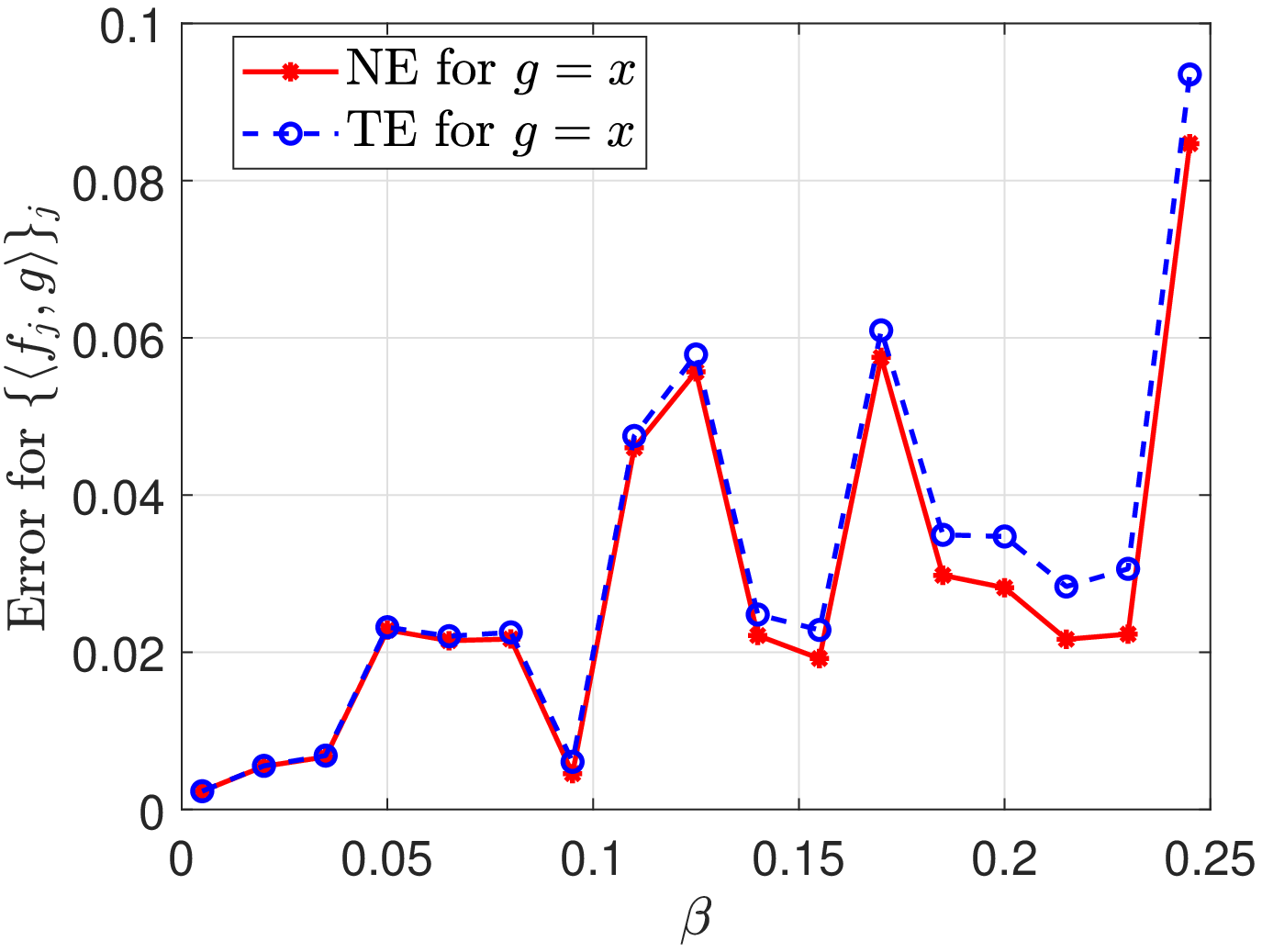} }
\subfloat
{\includegraphics[width=0.38\linewidth]{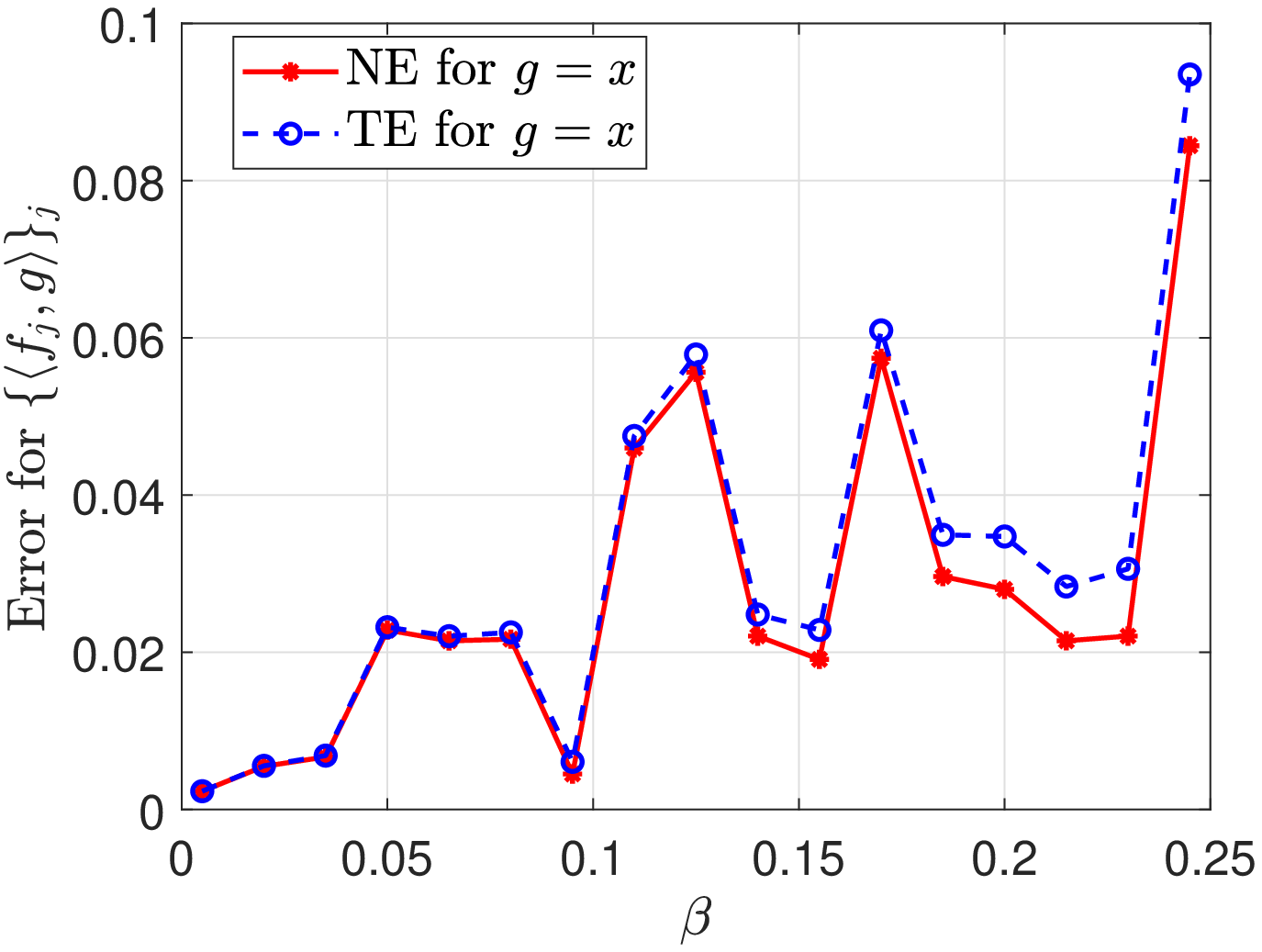} }
\hfill
\subfloat
{\includegraphics[width=0.38\linewidth]{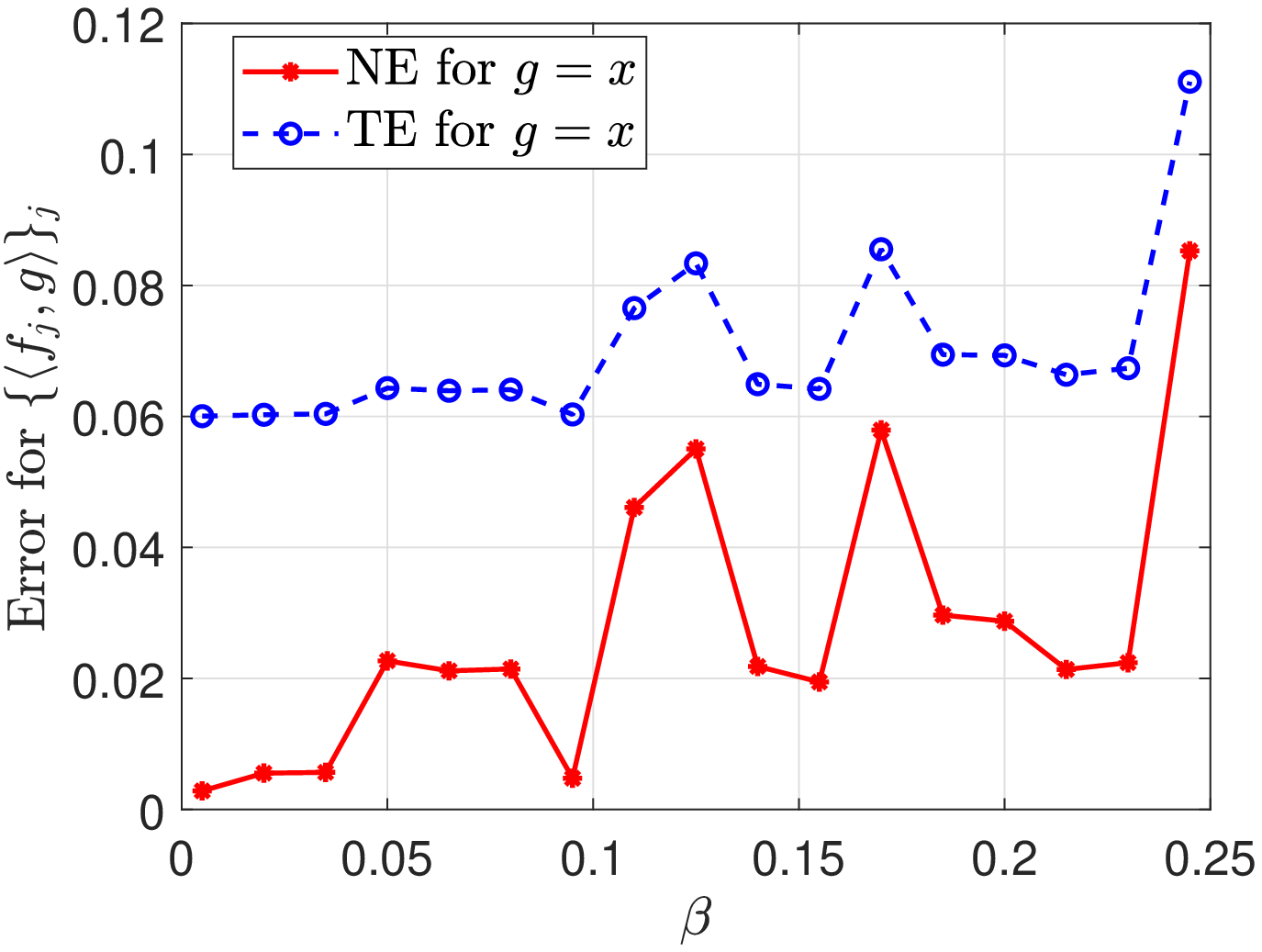}
}
\subfloat
{\includegraphics[width=0.38\linewidth]{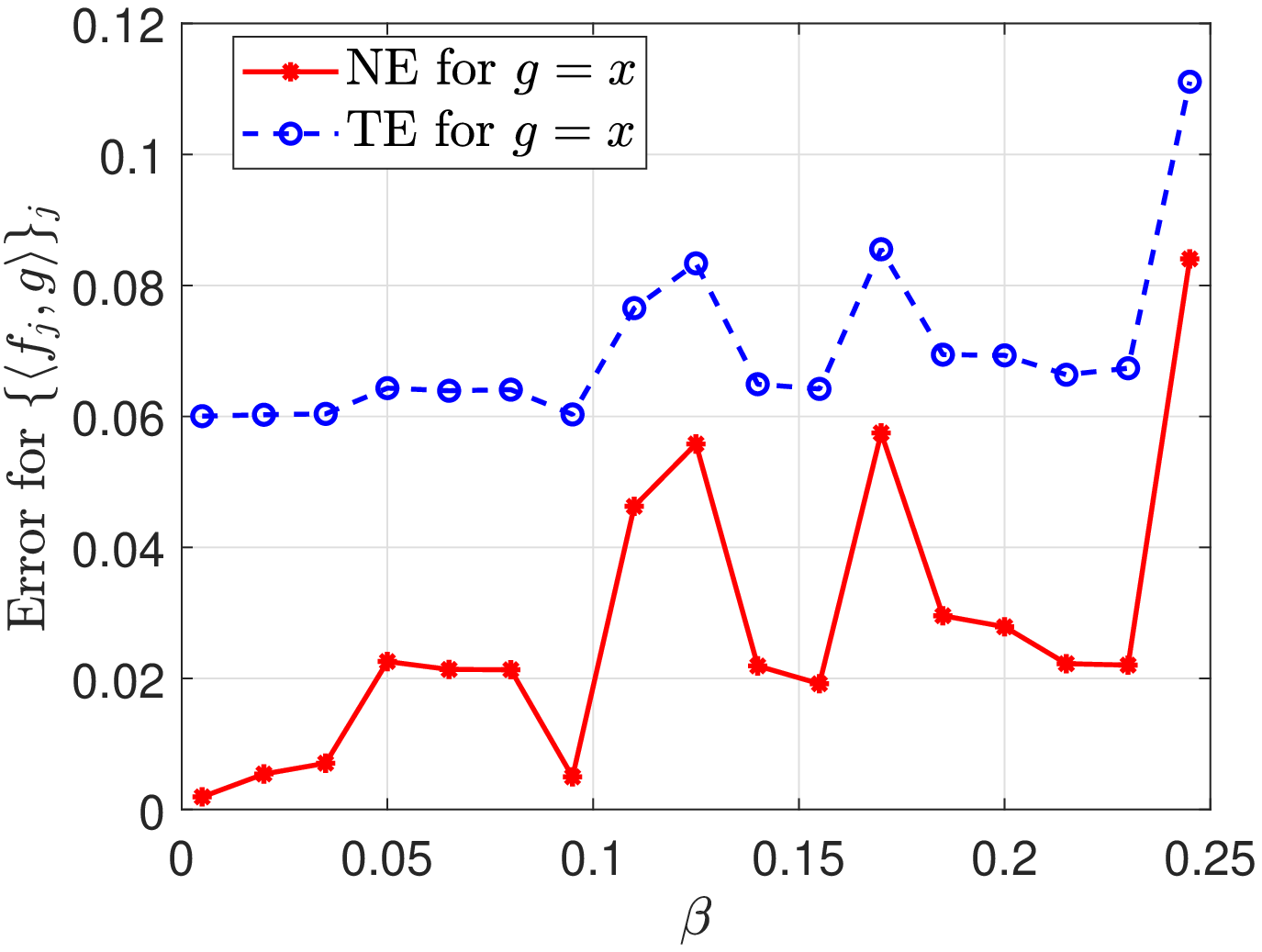}
}\hfill
\caption{ \footnotesize{The estimates of $\langle f_j,g\rangle$~\textit{vs.}~$\beta$ for Algorithm \ref{alg:direct}: $L=0.01$. \textbf{First column:}   $\eta(x,t)= \cos(Ltx)+C$. \textbf{Second column:}   $\eta(x,t)= x\exp(-Lt)+C$. The noise variances $\sigma$ on the measurements are $\sigma=0, 10^{-4}$ for the 1st, 2nd row, respectively.}}
    \label{fig:burst_vs_beta_direct}
\vspace{-0.05in}
\end{figure}

\begin{figure}[!h]
\vspace{-0.08in}
\centering
\subfloat
{\includegraphics[width=0.38\linewidth]{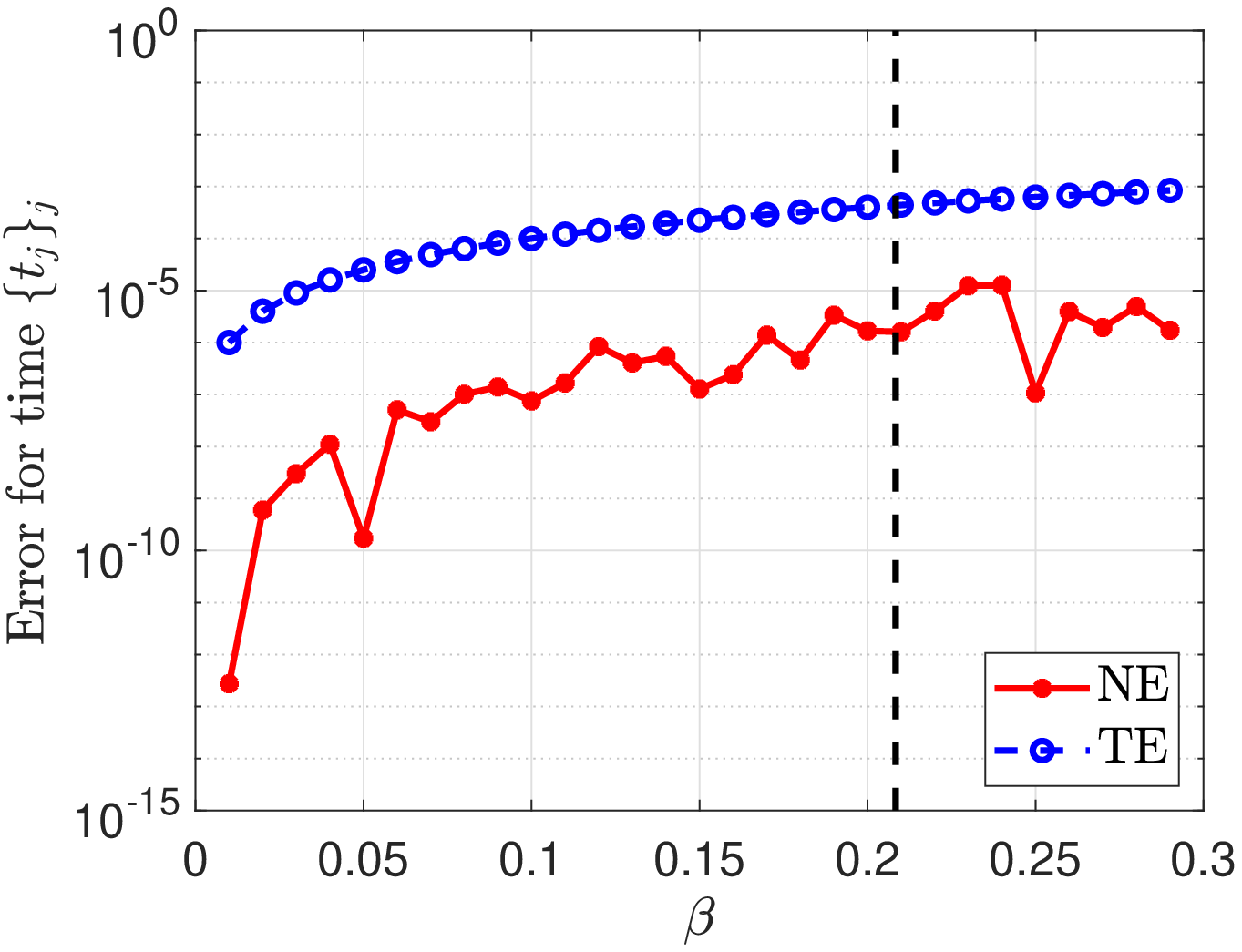}}
\subfloat
{\includegraphics[width=0.38\linewidth]{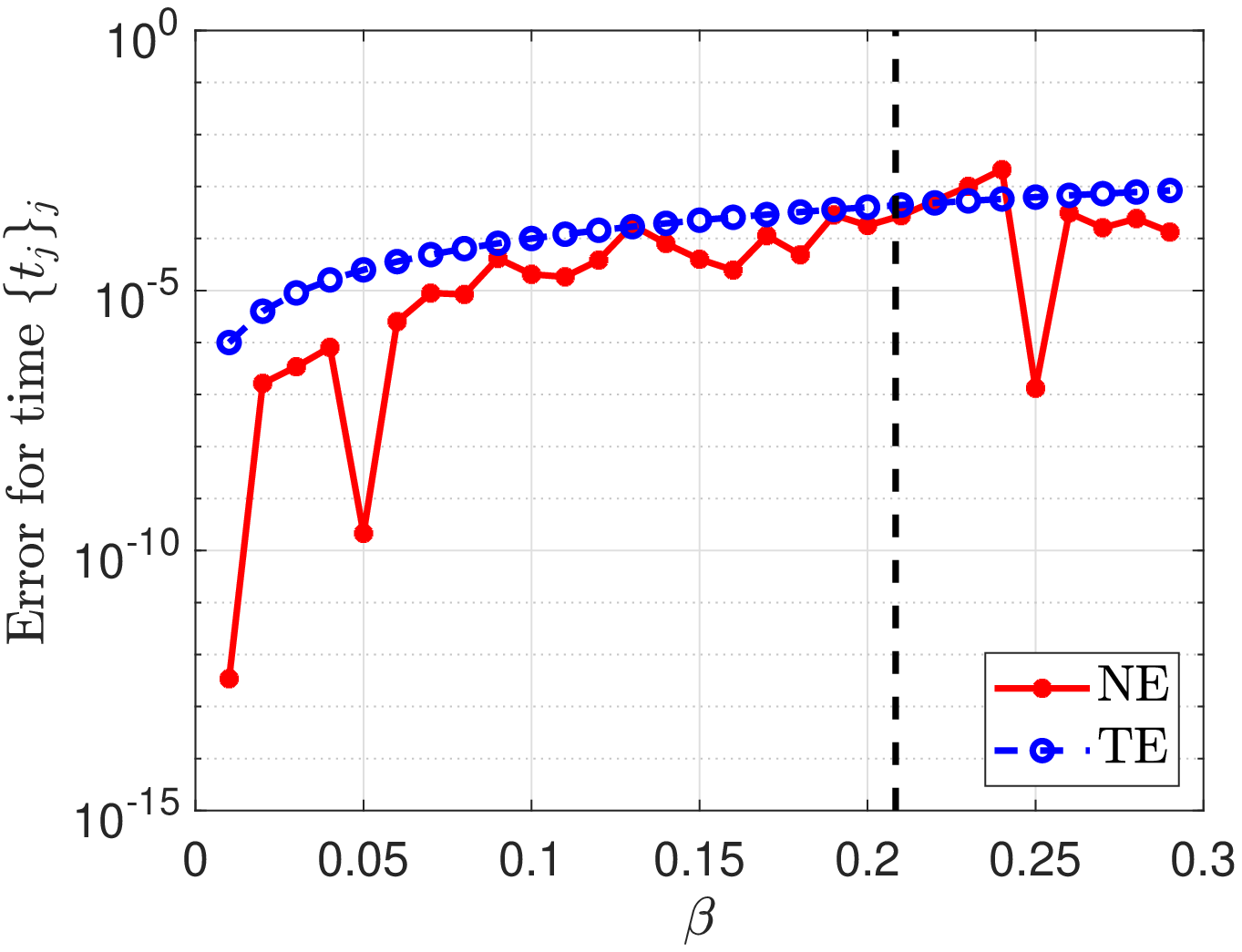}}\\

\subfloat
{\includegraphics[width=0.38\linewidth]{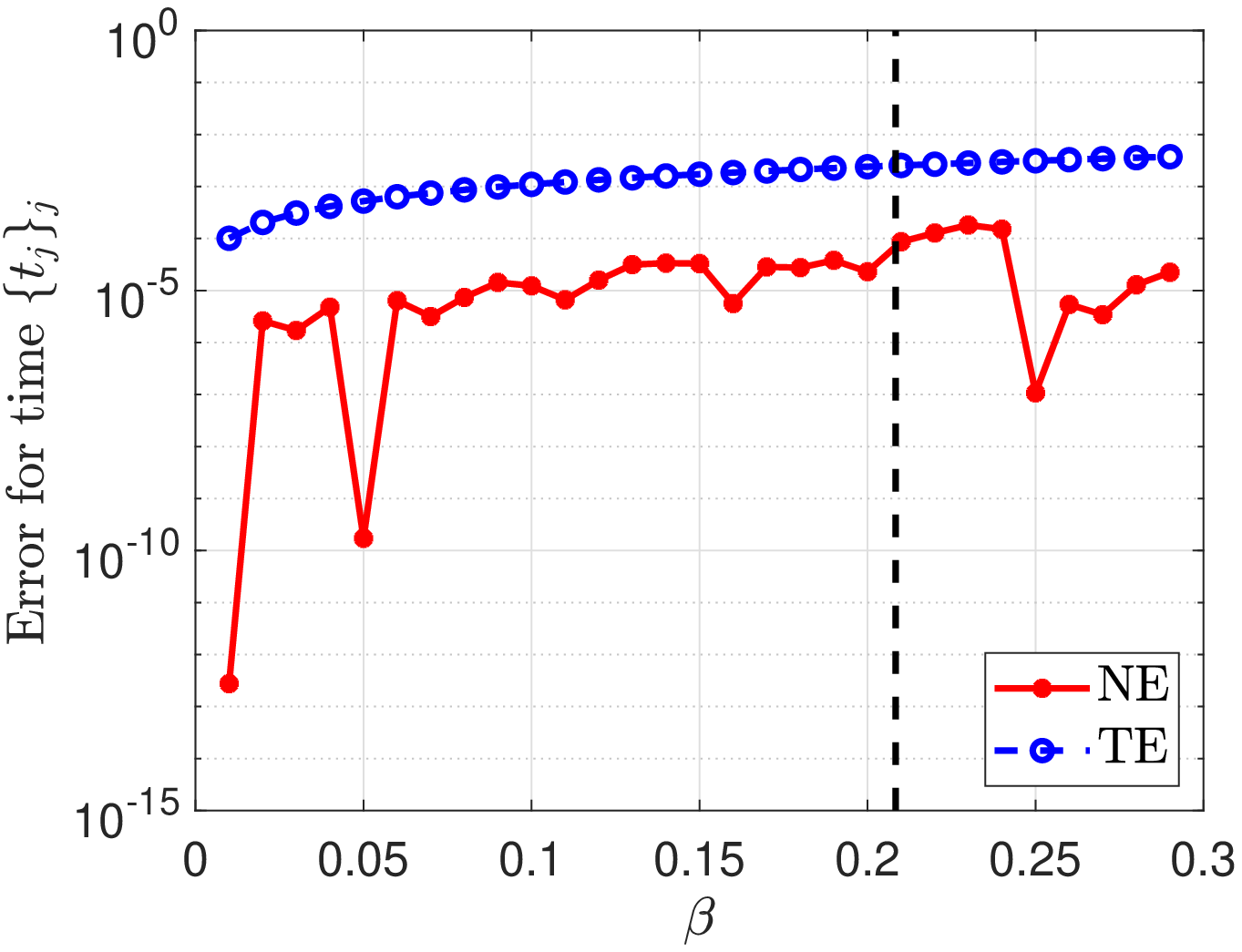}
}
\subfloat
{\includegraphics[width=0.38\linewidth]{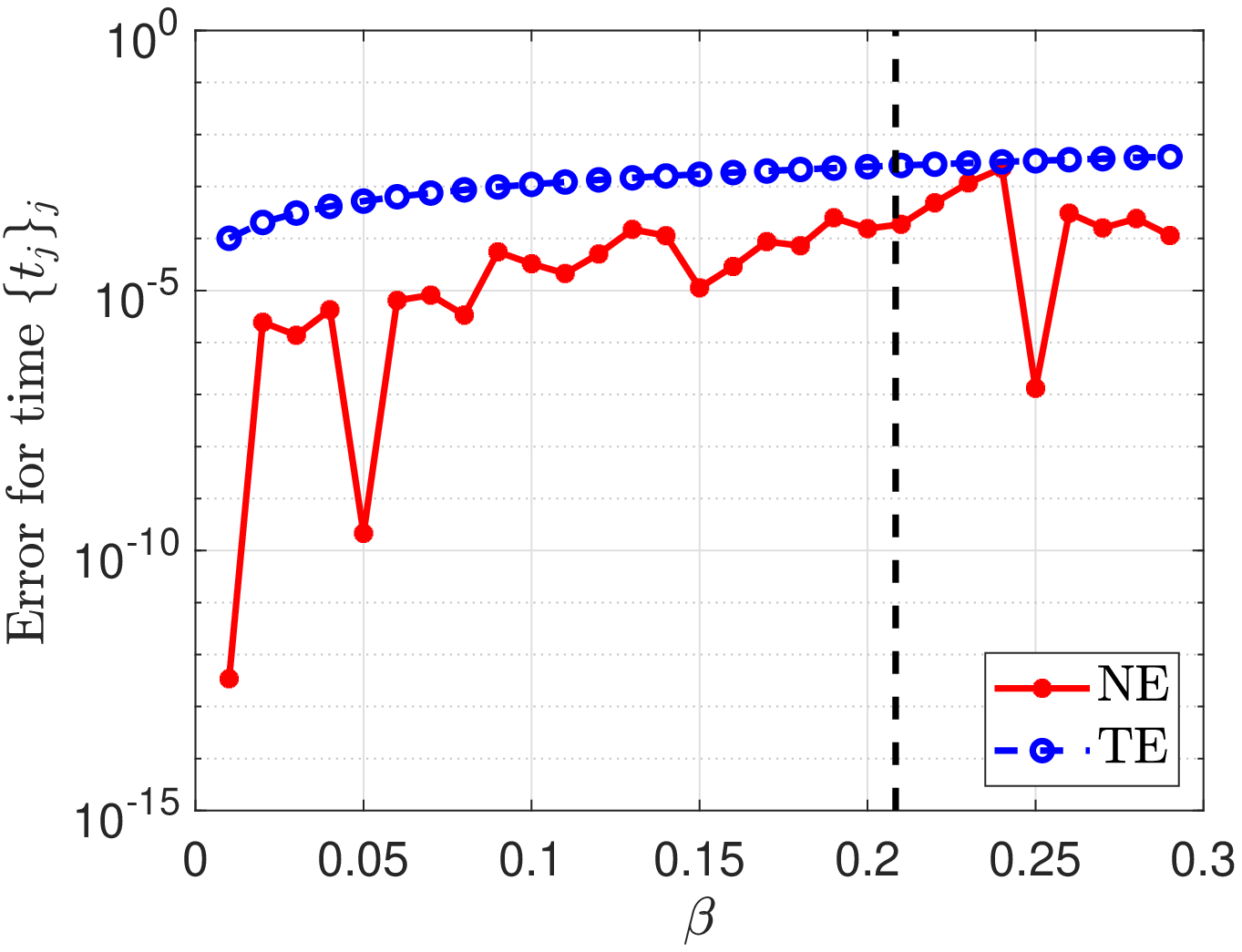}
}
\caption{ \footnotesize{The estimates of $t_j$~\textit{vs.}~$\beta$ for Algorithm \ref{alg2}: $L=0.01$. \textbf{First column:}  $\eta(x,t)= \cos(Ltx)+C$. \textbf{Second column:}   $\eta(x,t)= x\exp(-Lt)+C$. The noise variances $\sigma$ on the measurements are $\widetilde{\sigma}=0, 10^{-4}$ for the 1st, 2nd  row, respectively. NE and TE stand for the numerical error and the theoretical error respectively. The black vertical line represents the theoretical restrictions on $\beta$ see Assumption \ref{as4}. }}
    \label{fig:time_vs_beta_prony}
\vspace{-0.05in}
\end{figure}

\begin{figure}[H]
\vspace{-0.08in}
\centering
\subfloat
{\includegraphics[width=0.38\linewidth]{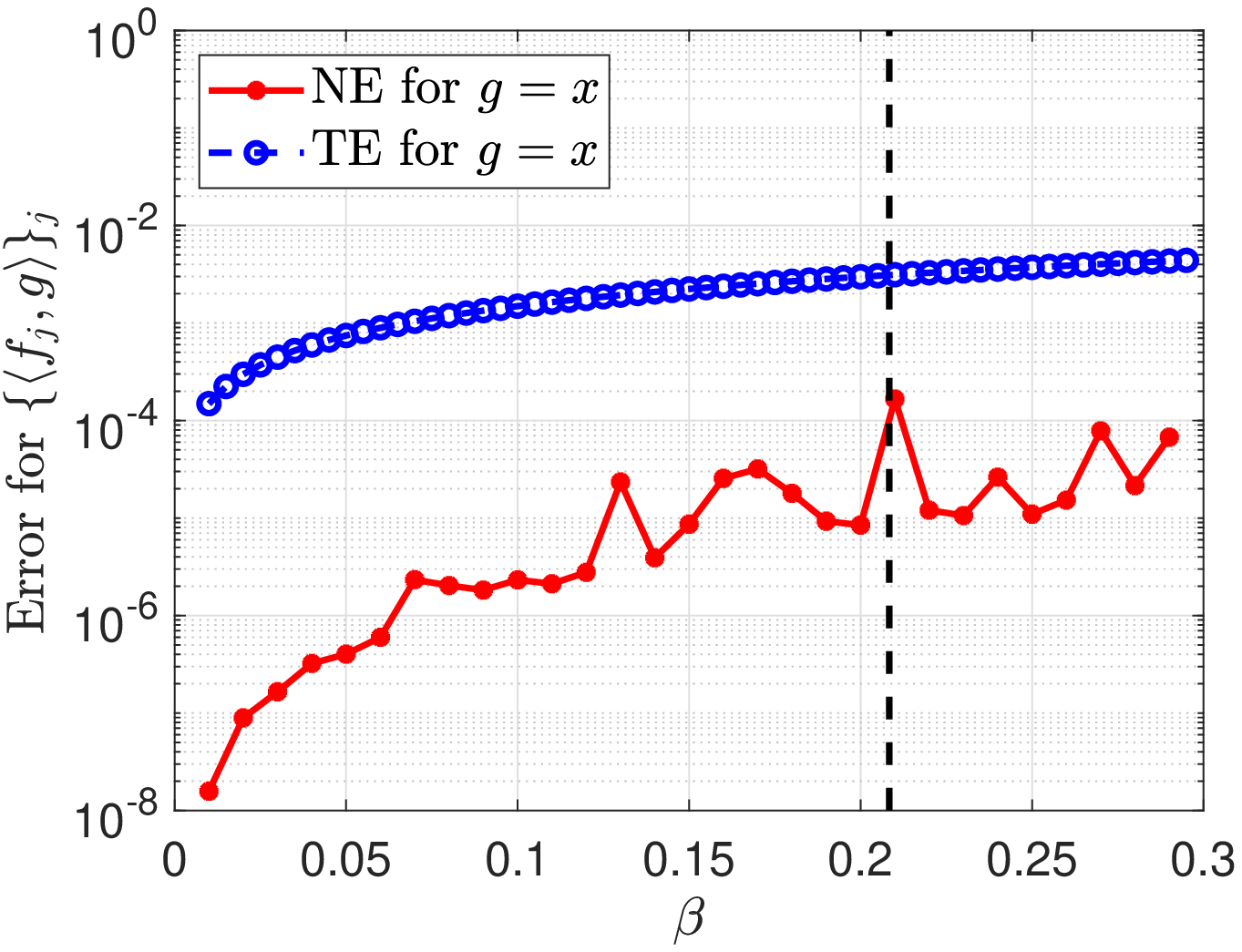} }
\subfloat
{\includegraphics[width=0.38\linewidth]{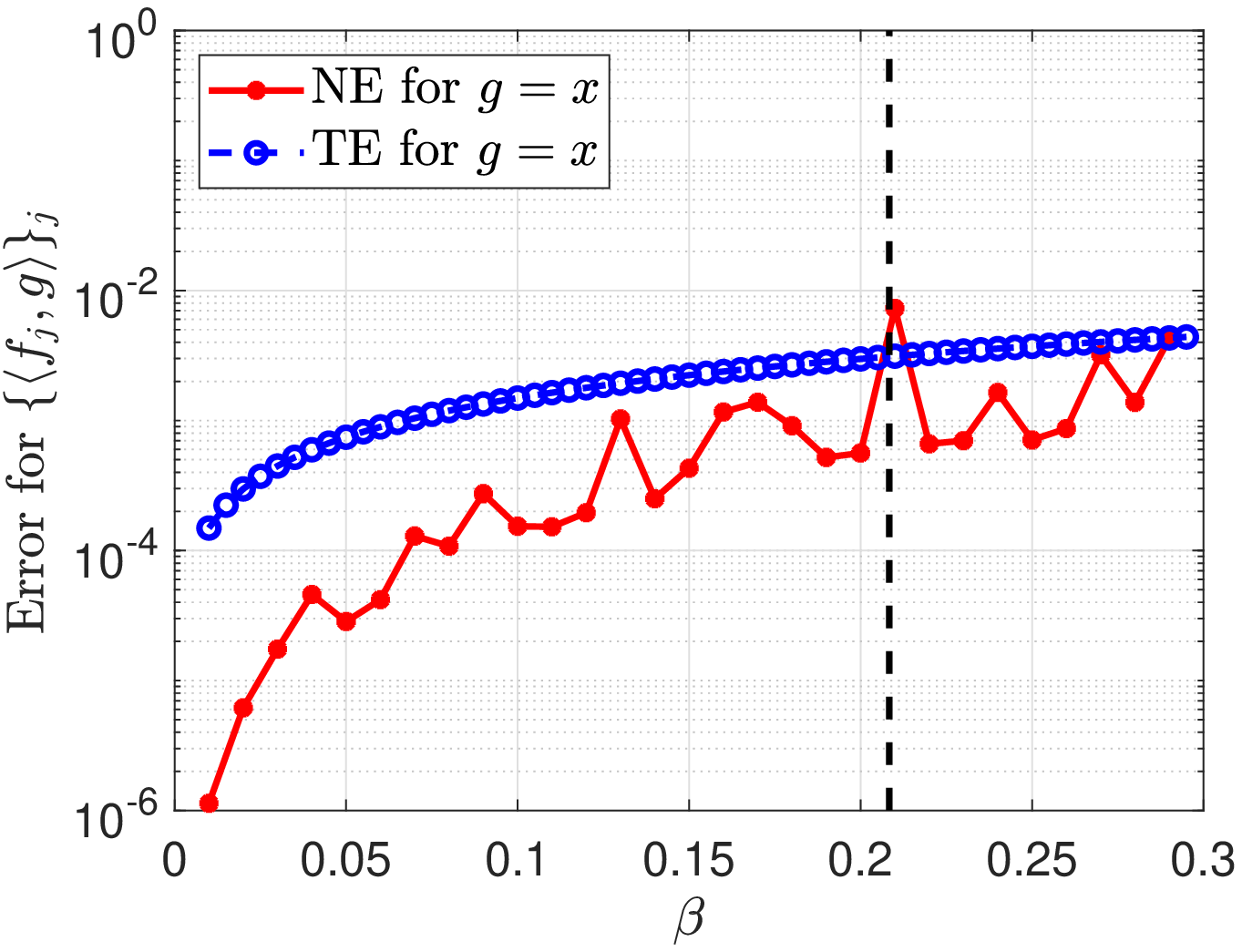} }
\hfill
\subfloat
{\includegraphics[width=0.38\linewidth]{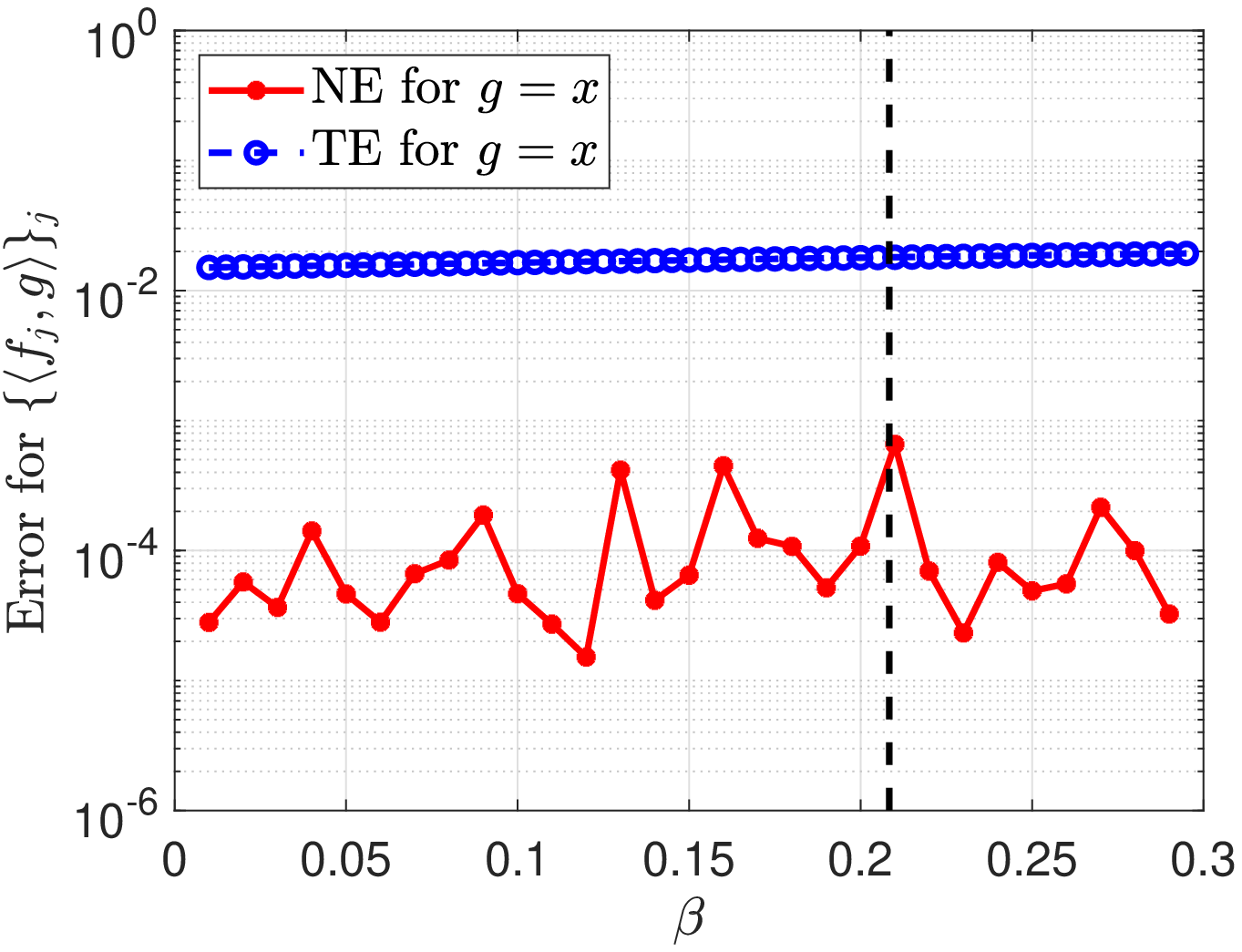}
}
\subfloat
{\includegraphics[width=0.38\linewidth]{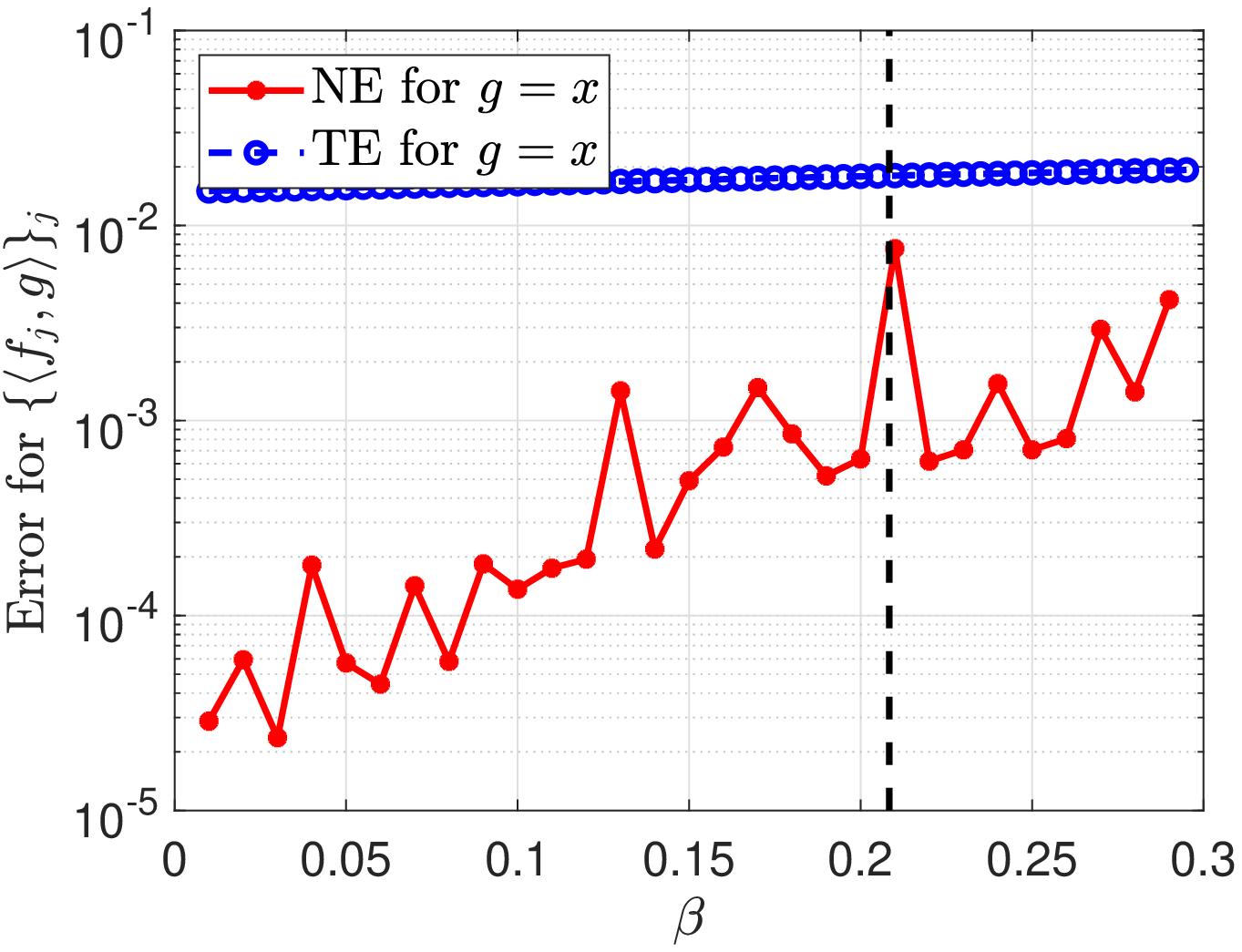}
}\hfill
\caption{ \footnotesize{The estimates of $\langle f_j,g\rangle$~\textit{vs.}~$\beta$ for Algorithm \ref{alg2}: $L=0.01$. \textbf{First column:}   $\eta(x,t)= \cos(Ltx)+C$. \textbf{Second column:}   $\eta(x,t)= x\exp(-Lt)+C$. The noise variances $\widetilde{\sigma}$ on the measurements are $\widetilde{\sigma}=0, 10^{-4}$ for the 1st, 2nd row, respectively. The black vertical line represents the theoretical restrictions on $\beta$ see Assumption \ref{as4}.}}
    \label{fig:burst_vs_beta_prony}
\vspace{-0.05in}
\end{figure}

 \begin{figure}[!h]
\vspace{-0.08in}
\centering
\subfloat
{\includegraphics[width=0.38\linewidth]{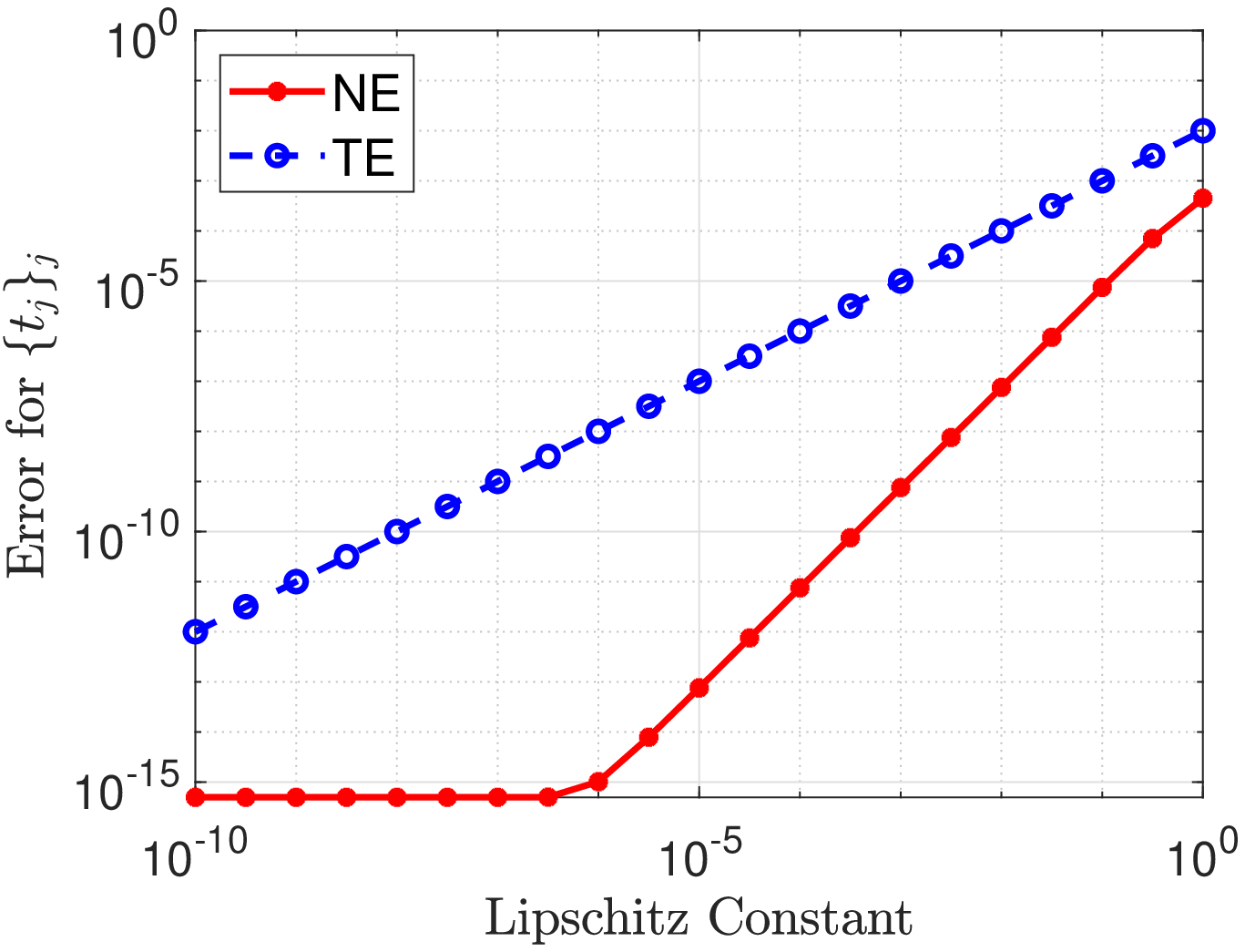}}
\subfloat
{\includegraphics[width=0.38\linewidth]{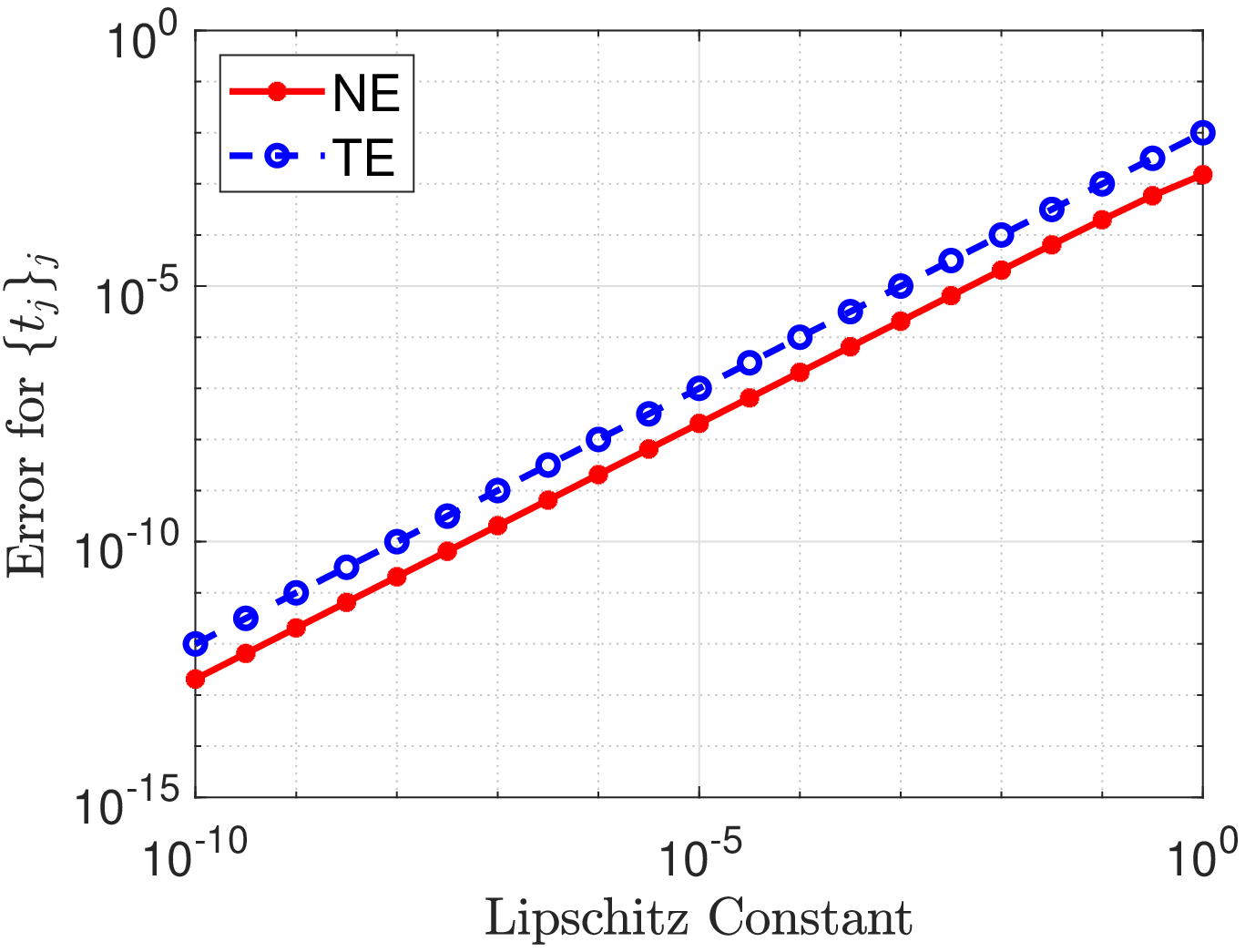}}\\
\subfloat
{\includegraphics[width=0.38\linewidth]{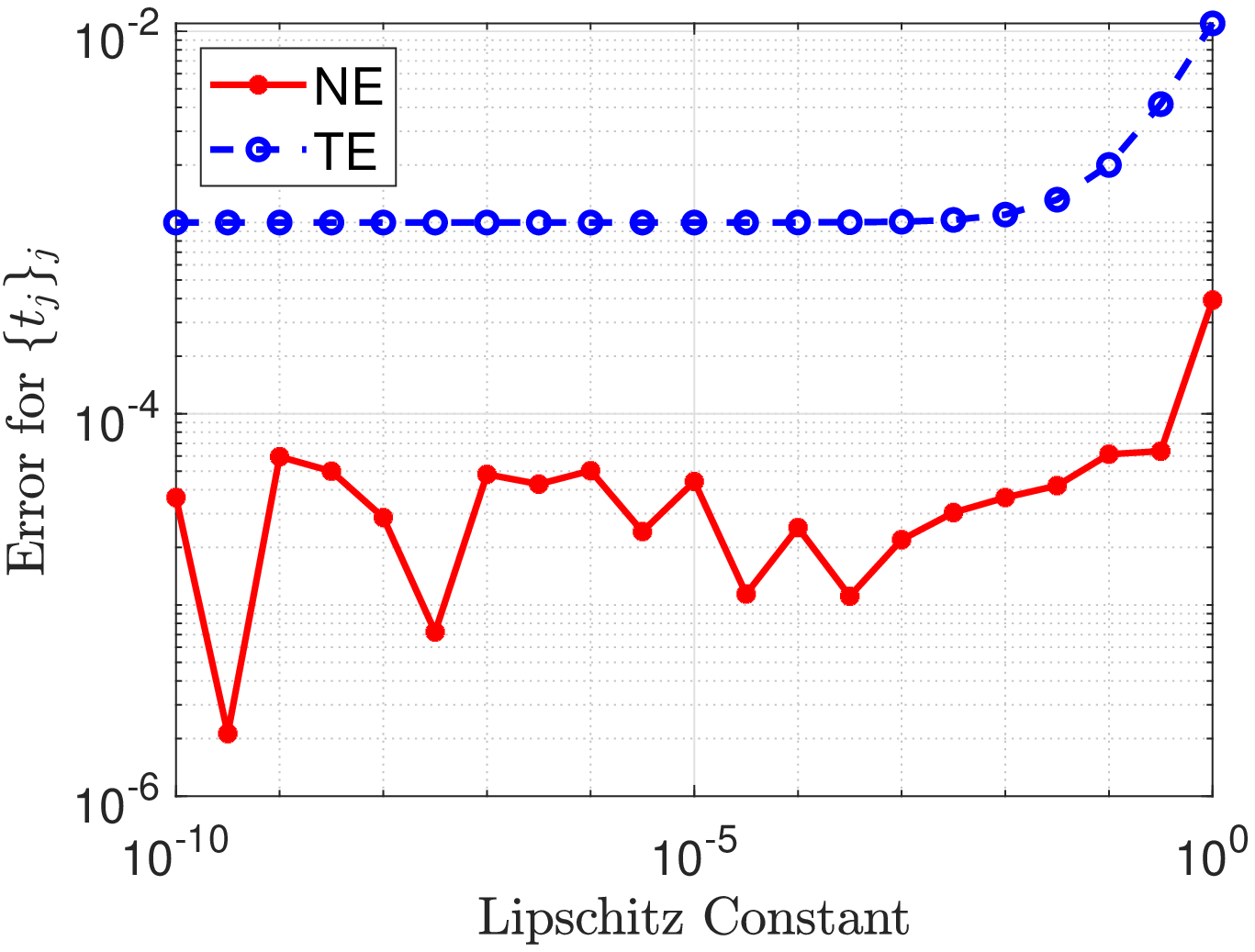}}
\subfloat
{\includegraphics[width=0.38\linewidth]{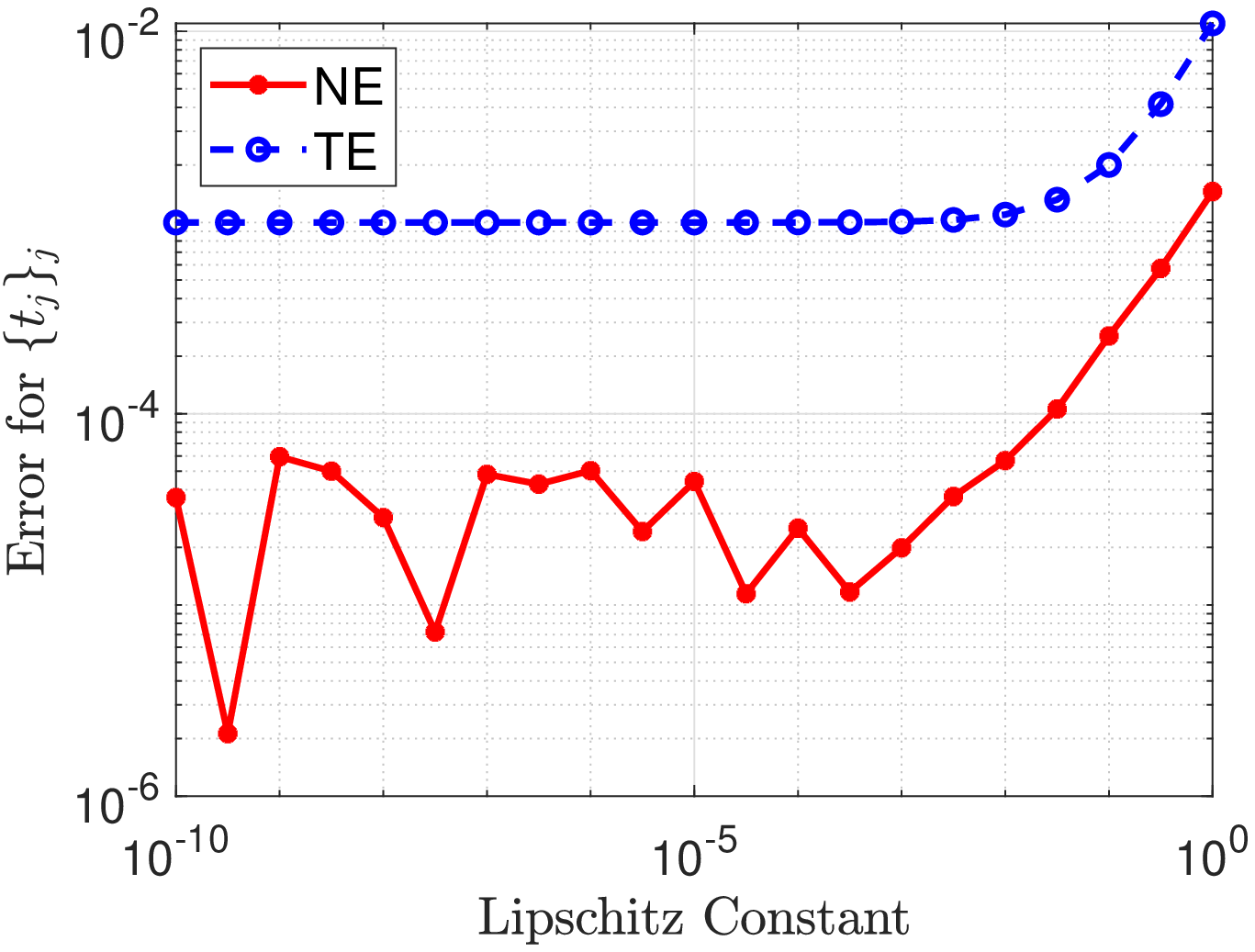}}
\caption{ \footnotesize{The estimates of $t_j$ \textit{vs.}  Lipschitz constant $L$  for Algorithm \ref{alg2}: $\beta=0.1$. \textbf{First column:} $\eta(x,t)= \cos(Ltx)+C$. \textbf{Second column:}  $\eta(x,t)= x\exp(-Lt)+C$. The noise variances $\widetilde{\sigma}$ on the measurements are $\widetilde{\sigma}=0,10^{-4}$ for the 1st, 2nd row, respectively. NE and TE stand for the numerical error and the theoretical error respectively.}}
    \label{fig:time_vs_lip_prony}
\vspace{-0.05in}
\end{figure}

\begin{figure}[H]
\vspace{-0.08in}
\centering
\subfloat
{\includegraphics[width=0.38\linewidth]{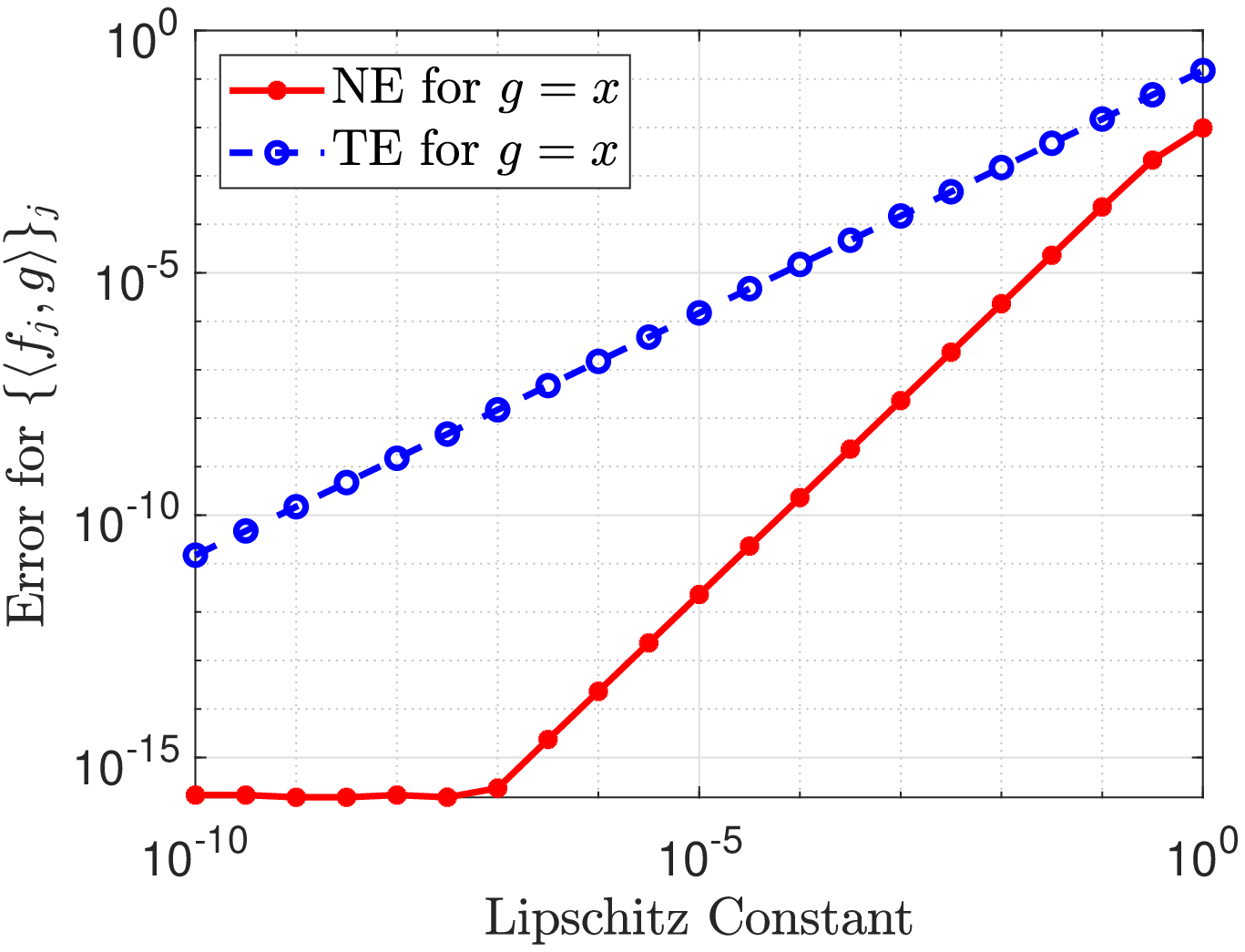}}
\subfloat
{\includegraphics[width=0.38\linewidth]{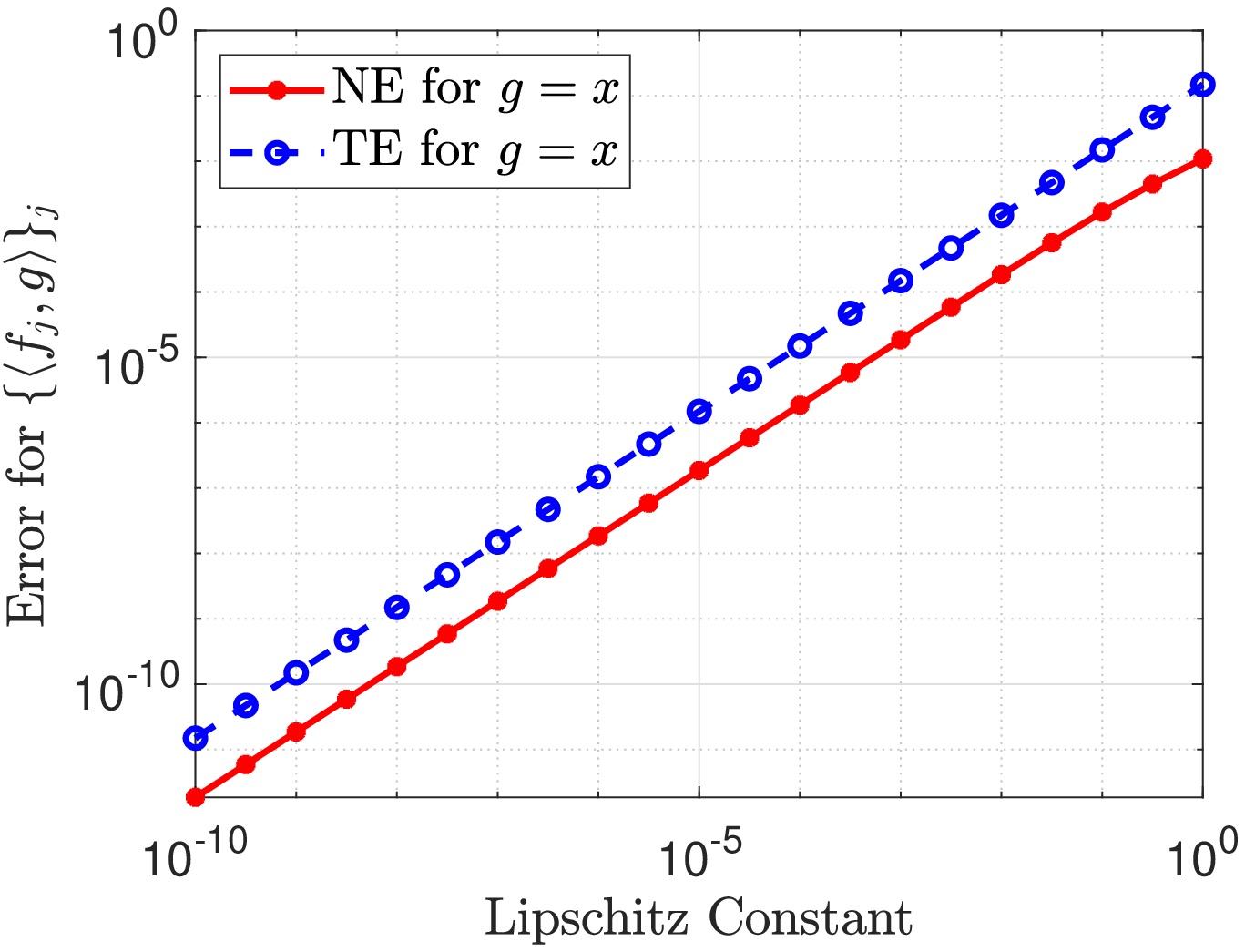}}
\\
\subfloat
{\includegraphics[width=0.38\linewidth]{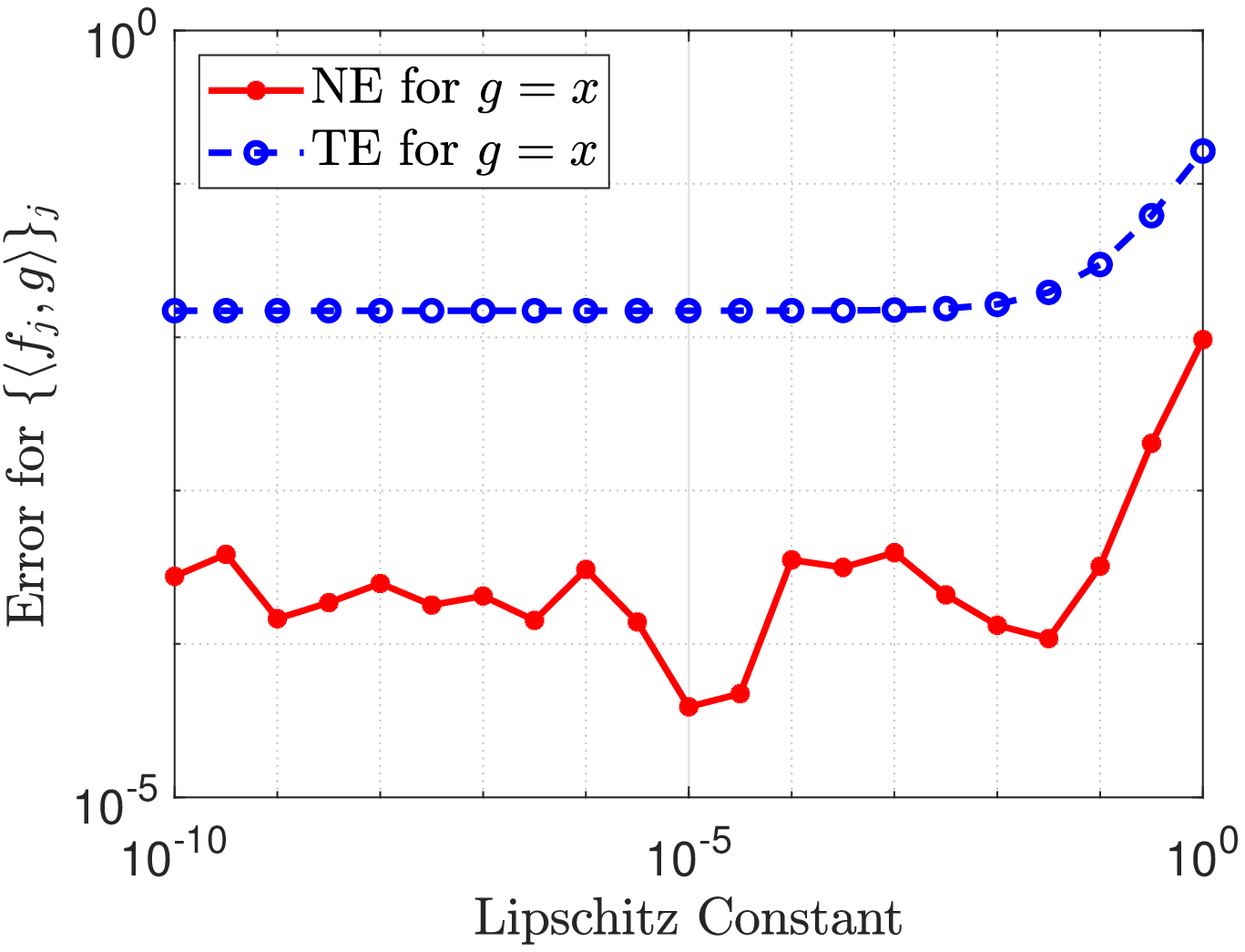}}
\subfloat
{\includegraphics[width=0.38\linewidth]{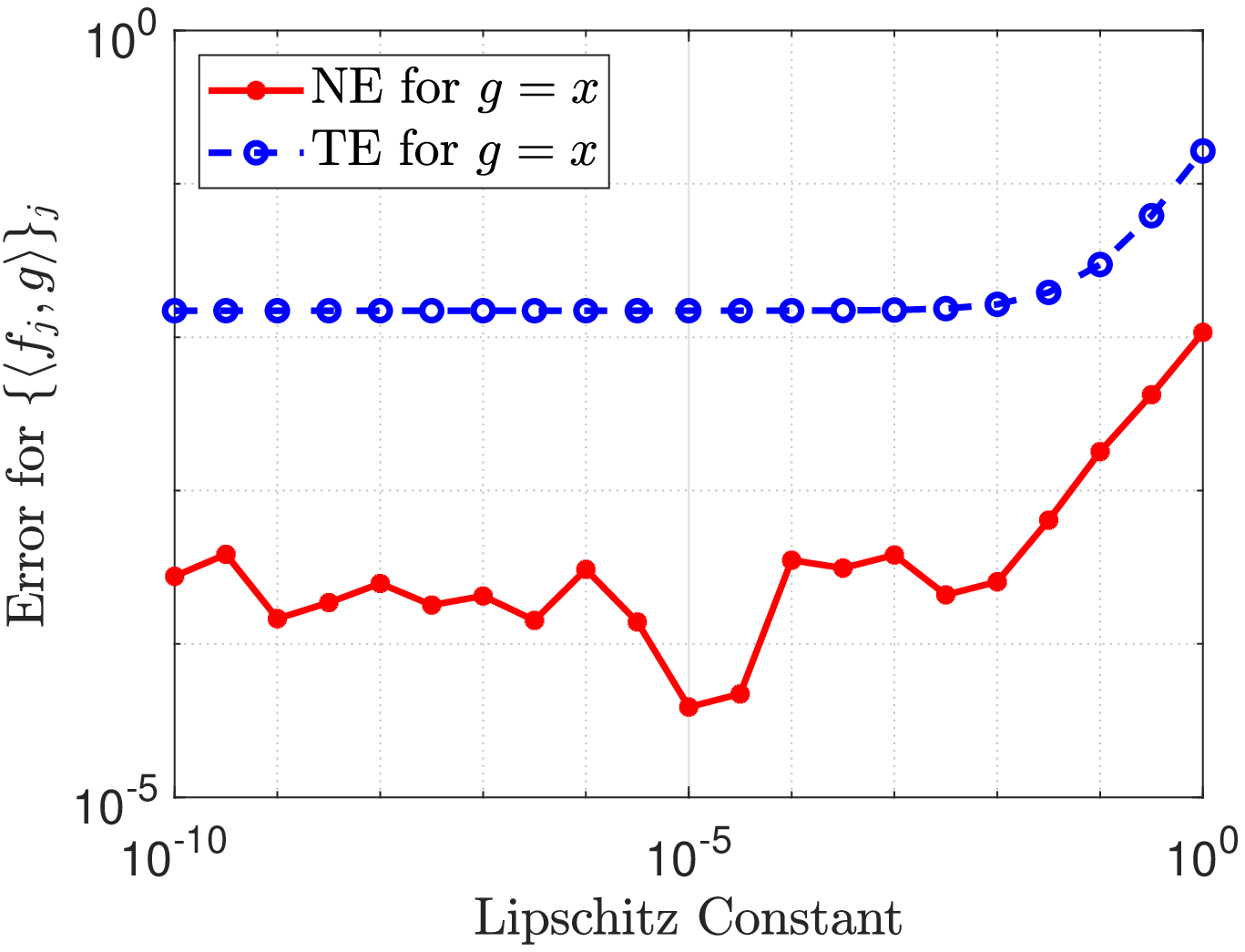}
}
\caption{ \footnotesize{The estimates of $\langle f_j,g\rangle$ \textit{vs.}  Lipschitz constant $L$ for Algorithm \ref{alg2}: $\beta=0.1$. \textbf{First column:}   $\eta(x,t)= \cos(Ltx)+C$. \textbf{Second column:}   $\eta(x,t)= x\exp(-Lt)+C$. The noise variances $\sigma$ on the measurements are $\widetilde{\sigma}=0,10^{-4}$ for the 1st, 2nd row, respectively.}}
    \label{fig:burst_vs_lip_prony}
\vspace{-0.05in}
\end{figure}

\section*{Acknowledgment}
 
The second author was supported in part by  NSF DMS \#2011140 and the Dunn Family Endowed Chair fund. The third author was supported in part by Simons Foundation Collaboration grant  244718.  

\bibliographystyle{siam}
\bibliography{refs}
 
\end{document}